\documentclass[journal ,onecolumn]{IEEEtran}
\IEEEoverridecommandlockouts
\usepackage{graphicx,cite,calc,dsfont}
\usepackage[defs]{ams}
\usepackage{eqnarray,subeqn,xcolor}
\usepackage{hyperref}
\usepackage{amsthm,amssymb}
\usepackage{algorithm}
\usepackage{algpseudocode}
\newtheorem{thm}{Theorem}
\newtheorem{lem}{Lemma}
\newtheorem{cor}{Corollary}
\theoremstyle{definition}
\newtheorem{examp}{Example}

\newtheorem{myremark}{Remark}

\begin{document}

\title{Multi-Server Coded Caching}

\author{ Seyed Pooya Shariatpanahi$^1$ , Seyed Abolfazl Motahari$^{2,1}$, Babak Hossein Khalaj$^{3,1}$ \\[4mm] 

1: School of Computer Science, Institute for Research in Fundamental Sciences (IPM), Tehran, Iran. \\
2: Department of Computer Engineering, Sharif University of Technology, Tehran, Iran. \\
3: Department of Electrical Engineering,
Sharif University of Technology, Tehran, Iran.\\
(emails: pooya@ipm.ir, \{motahari,khalaj\}@sharif.edu)\\ }

\maketitle

\thispagestyle{empty}
\pagestyle{empty}

\begin{abstract}
In this paper, we consider multiple cache-enabled clients connected to multiple servers through an intermediate network. We design several topology-aware coding strategies for such networks. Based on topology richness of the intermediate network, and types of coding operations at internal nodes, we define three classes of networks, namely, dedicated, flexible, and linear networks. For each class, we propose an achievable coding scheme, analyze its coding delay, and also, compare it with an information theoretic lower bound. For flexible networks, we show that our scheme is order-optimal in terms of coding delay and, interestingly, the optimal memory-delay curve is achieved in certain regimes. In general, our results suggest that, in case of networks with multiple servers, type of network topology can be exploited to reduce service delay. 

\let\thefootnote\relax\footnotetext{This research was in part supported by a grant from IPM, and by a grant from Iran National Science Foundation under Grant 92017806.}
\end{abstract}

\section{Introduction}\label{Sec_Intro}
Unprecedented growth in transmit data volumes throughout the networks in recent years demands more efficient use of storage devices while providing high quality of service (QoS) to the users. Currently, large files are stored on servers and users' requests are stored in queues waiting to get service from them. Naturally, one approach to reduce congestion in such networks is to increase the service rate of such servers. However, this will put additional burden on such nodes. As the cost of storage devices has decreased over the years, another viable option is to provide geographical content replication in the network through use of the so-called low-capacity \emph{caching} nodes. The idea of using such nodes for data replication and providing easier local access to data is already covered in the literature (see for example \cite{Kangasharju, Baev_2008, Dowdy_2012, Podlipnig, Borst_2010, Gitzenis_2013}).
 
Recently, in their seminal work, Maddah-Ali and Niesen considered a single server network and have shown that through a two-phase cache placement and content delivery strategy, server load can be reduced inversely proportional to the total size of cache introduced in the network. In fact, in the cache placement phase, contents are stored on caches without knowing the actual demands of the users and in content delivery phase the server transmits packets to fulfill the demands. The fact that such \emph{global caching gain} can be achieved in such network is surprising as  the demands are not known apriori at the cache placement phase.

Maddah-Ali and Niesen's cache placement strategy is based on shattering each file into many pieces and only distributing them throughout the caching nodes \emph{without} replication. It should be noted that such approach is in contrast to the conventional local cache placement strategies where a file or a single piece of it is replicated in caches. The astounding feature of their strategy is that transmission of a single packet at content delivery phase can then simultaneously serve several users. Imagine that two pieces of two files are stored at two different caches and each of them requires the piece available at the other. A single packet containing the sum of two packets can be sent to fulfill both users' demands.  They have shown that their strategy is 12-approximation of the optimal strategy.

The network considered in \cite{Maddah-Ali_Fundamental_2014} is a simple broadcast network where a packet transmitted by the server arrives unaltered at all users. A fundamental problem is to see how network topology affects the optimal coding strategy through both placement and delivery phases.

One of the simplest topologies is the tree network. In \cite{Maddah-Ali_Decentralized_2014}, Maddah-Ali and Niesen proved that their original strategy can be used directly for such a network and what is needed to achieve 12-approximation of the optimal strategy is a simple topology-aware routing strategy at the internal nodes; An internal node routes a packet on its output port if the packet is useful for at least one of the port's children. 

While the topology-aware routing scheme for tree networks is shown to be an order-optimal solution, real-world topologies are much more sophisticated than the simple tree structure. In this paper, we characterize the effect of network topologies on code design and performance analysis of coded caching in a more general setup. In particular, we investigate a multi-server network topology where a set of servers are connected to the clients through an error-free and delay-free intermediate network of nodes (see Fig. \ref{Fig_Multi_Server_Model}). We assume that each node in the intermediate network can perform any causal processing on its input data, to generate its outgoing data. This can consist of simple routing or more sophisticated network coding schemes.

The objective considered in \cite{Maddah-Ali_Fundamental_2014} is minimizing the traffic load imposed to the single server. However, in general, other objectives may be of higher importance when designing network operation strategy. One such key criterion is the \emph{service delay} of the network which is specially critical in content delivery networks (see eg. \cite{Chen_2002,Vakali_2003,Niesen_Delay_2014}). We define the service delay of the network as the total time required to serve any given set of the clients' requests for a specific strategy. We distinguish between two types of delay, \emph{Network Delay} $T_N$ and \emph{Coding Delay} $T_C$, where the total service delay, $T$, is given by $T=T_N+T_C$. To be more precise, $T_N$ is the time it takes for packets to be routed through the network and arrive at their requesting nodes. Naturally, $T_N$ mainly captures the links and queues delays in the network which are intrinsic characteristics of the network. On the other hand, $T_C$ captures the transmission block length required to serve all the users for a specific coding strategy. In this paper, we focus on the coding delay and design strategies to minimize such delays.

We consider three classes of networks: 1- dedicated networks, 2- flexible networks, and 3- linear networks. These networks are characterized based on the richness of their internal connections, as shown in Fig. \ref{Fig_Model_Example}. In each class, an important network topology aspect is the number of servers connected to the network, and their points of contact. In dedicated networks, we can dedicate each server to serve a fixed subset of clients, where each server can send a common message to its corresponding subset, interference-free from other servers. Although in dedicated networks the assignment of clients to the servers is fixed, in flexible networks the network topology is rich-enough to let us adapt these assignments during network operation. Finally, in linear networks we assume random linear network coding operations at the internal nodes. Consequently, in linear networks, the network input-output relation is characterized by a random matrix. As we show in this paper, in order to minimize the coding delay, designing the coding strategy for each class should carefully utilize the flexibility of that class. As will be shown subsequently, there exist coding strategies outperforming that of \cite{Maddah-Ali_Decentralized_2014} for all of the three classes of  networks. Interestingly, we obtain an order optimal solution for the flexible networks. 

Finally, let us review some notations used in this paper. We use lower case bold-face symbols to represent vectors, and upper case bold-face symbols to represent matrices. For any matrix $\mathbf{A}$, $\mathbf{A}^t$ denotes the transpose of $\mathbf{A}$ and for any vector $\mathbf{a}$, $\mathbf{a}^{\perp}$ shows that the condition $\mathbf{a}.\mathbf{a}^{\perp}=0$ is satisfied. For any two sets $S_1$ and $S_2$, the set $S_1 \backslash S_2$ consists of those elements of $S_1$ not present in $S_2$. Also we define $[K]=\{1,\dots,K\}$ and $\mathbb{N}$ to be the set of integer numbers. Moreover, $\mathbb{F}_q$ shows a finite field with $q$ elements, and $\mathbb{F}_q^{a \times b}$ denotes the set of all $a$-by-$b$  matrices whose elements belong to $\mathbb{F}_q$. Finally, let $x_1,\dots,x_m \in \mathbb{F}_q$, then $L(x_1,\dots,x_m)$ is a random linear combination of $x_1,\dots,x_m$ where the random coefficients are uniformly chosen from $\mathbb{F}_q$. 

The rest of the paper is organized as follows. In Section \ref{Sec_Model}, we describe the network model and different classes of networks. In Section \ref{Sec_Main_Results}, we review the main results of the paper, present some examples, and discuss their implications. The next two sections, i.e. Sections \ref{Sec_Flexible} and \ref{Sec_Linear}, present the details of the coding strategies proposed for flexible and linear networks, respectively. Finally, we conclude the paper in Section \ref{Sec_Conclusions}.

\section{Model and Assumptions}\label{Sec_Model}

\begin{figure}
\begin{center}
\includegraphics[width=0.4\textwidth]{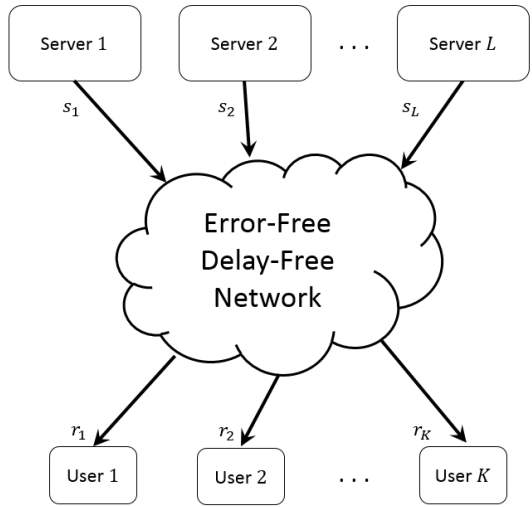}
\end{center}
\caption{Network Model. \label{Fig_Multi_Server_Model}}
\end{figure}

Consider $L$ servers connected to $K$ users through a network. By network we mean a \textit{Directed Acyclic Graph} (DAG) $\mathcal{G}=(V,E)$, in which the set of vertices $V$ consists of internal nodes, and every edge $e\in E$ on the graph represents an \textit{error-free} and \textit{delay-free} link with capacity of one symbol per channel use. Each server and each user is connected to the network by a single link with capacity of one symbol per channel use. At each channel use each node inside the network sends symbols on its output links based on (deterministic/random) functions  of the symbols on its input links, without introducing any delay, where functions corresponding to different output ports need not be the same. Also, we assume that there is no inter-link interference. Data is represented by $m$-bit \textit{symbol}s which are members of a finite field $\mathbb{F}_{q}$, where $q=2^m$.  

Consider a library of $N$ files $\{W_1,\dots,W_N\}$ each of $F$ bits is available to all servers. Each user is also assumed to have a cache of size $MF$ bits. During its operation, the network experiences two different traffic conditions, namely \emph{low-peak} and \emph{high-peak} leading to different network transmission costs for the two conditions. Based on the given traffic condition, the network operates in two distinctive phases. The first phase that is performed during low-peak condition is called the \emph{cache content placement} phase at which servers send data to the users  without knowing the actual requests of the users. This data is cached at the users with the size constraint of $MF$ bits and is stored to be used in the future. In the second phase that is performed during high-peak network condition, each user requests one of the files (demand $d_k$ of user $k$ denotes requesting file $W_{d_k}$), and according to these requests the servers send proper packets over the network. Subsequently, upon receipt of packets over the network, users try to decode their requested files with the help of their own cache contents. Assuming that the cache placement transmission delay during the low-peak condition puts no constraint on overall network performance, the goal is to design the cache placement strategy such that the service delay at the time of \emph{content delivery} is minimized.

Channel uses in the network are indexed by time slots $t=1,2,\dots$. At time slot $t$, servers transmit symbols $s_1(t),\dots,s_L(t)$ and users receive symbols $r_1(t), \dots, r_K(t)$ without delay. We consider the most general case, i.e.,
\begin{eqnarray}
\nonumber r_k(t)=f_k(s_1(t),\dots,s_L(t)), k=1,\dots,K,
\end{eqnarray}
in which we have assumed that the network is memory-less across time slots. Functions $f_k(.)$ depend on the topology of the network and the local operations of the nodes inside the network.

We define:
\begin{eqnarray}
\nonumber
\mathbf{s}(t) \triangleq \left( \begin{array}{c} s_1(t) \\ \vdots \\ s_L(t)   \end{array} \right), 
\mathbf{r}(t) \triangleq \left( \begin{array}{c} r_1(t) \\ \vdots \\ r_K(t)   \end{array} \right),                       
\end{eqnarray}
where $\mathbf{s} \in  \mathbb{F}_{q}^{L \times 1}$, and  $\mathbf{r} \in  \mathbb{F}_{q}^{K \times 1}$.

In the first phase, users store data from the servers without knowing the actual requests. The only concern in the first phase is respecting the memory constraint of each user. However, in the second phase, we focus on the time needed to deliver the requested files to the users. The second phase consists of $T_C$ time slots (channel uses). In other words, during the second phase, servers sequentially transmit $\mathbf{s}(1), \mathbf{s}(2), \dots, \mathbf{s}(T_C)$, and the users receive $\mathbf{r}(1), \mathbf{r}(2), \dots, \mathbf{r}(T_C)$. Consequently, $T_C(d_1,\dots,d_K)$ is the number of times slots required to satisfy demands $d_1,\dots,d_k$. Then, we define the optimum \emph{Coding Delay} as:
\begin{eqnarray}
D^*=\min {\max_{d_1,\dots,d_k}{T_C(d_1,\dots,d_K)}}, 
\end{eqnarray}
where the minimization is over all strategies. In this paper, we are interested in characterizing $D^*$ for a network, given its specific topology.

For a given network topology, the network input-output relation depends on operational design of internal nodes. As we will show, the \emph{richer} the network topology is, the broader the design space will be. Therefore, we consider the following three classes of network topologies:

\begin{figure}
\begin{center}
\includegraphics[width=1\textwidth]{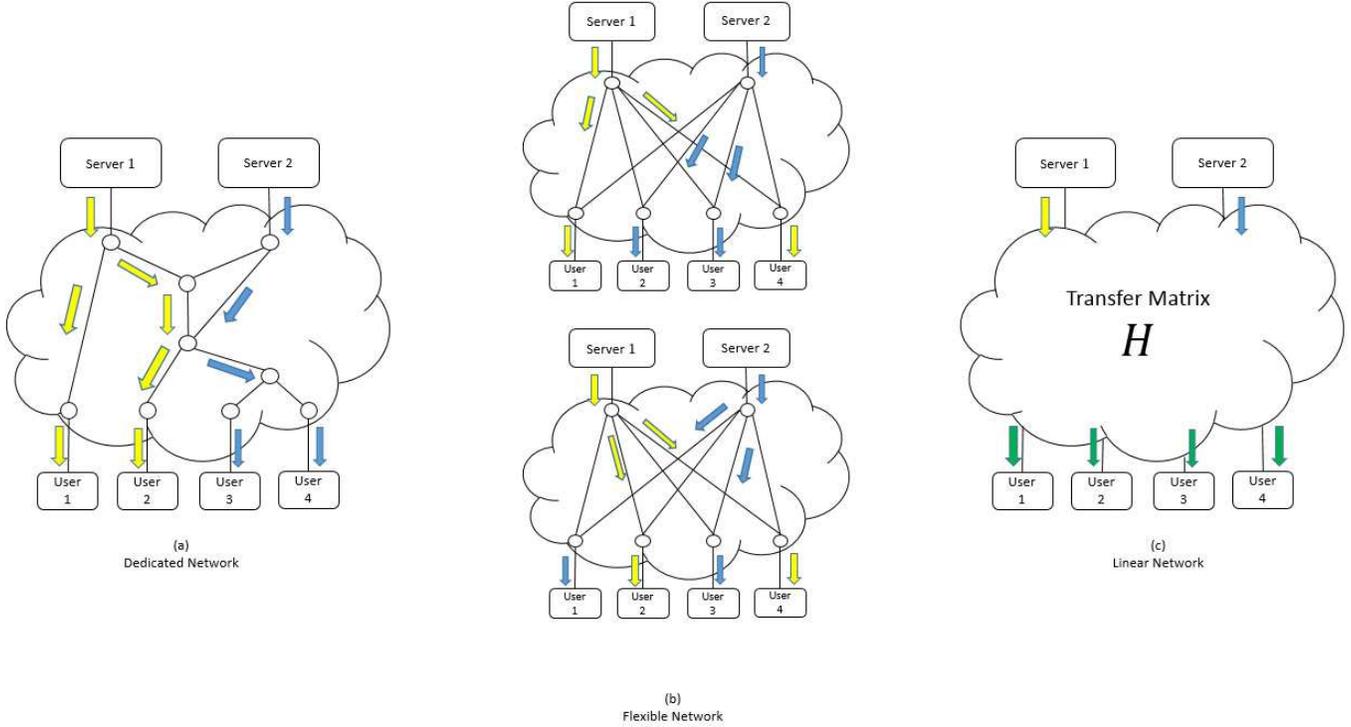}
\end{center}
\caption{Examples for dedicated, flexible, and linear networks.\label{Fig_Model_Example}}
\end{figure}

\begin{itemize}

\item \textbf{Dedicated Networks}

In this class of networks, each packet transmitted by a server is routed to a fixed subset of the users. In other words, we can dedicate each server to a fixed subset of users, and this server can send packets to these users, concurrently and without interference to other servers. We assume these subsets to be non-overlapping so that each user is assigned to a single server. Also, we assume we can balance these assignments such that the number of users assigned to a server is almost the same for all servers. If network topology allows us to perform such assignments, we call the network a \emph{Dedicated Network}.

More formally, there exists a coding (in this case just routing would suffice) strategy at the network nodes such that there \textit{exists a partitioning} $\{P_1,\dots,P_L\}$ of $[K']=\{1,2, \dots,K'\}$ where
\begin{eqnarray}
 \nonumber & &|P_l|=\frac{K'}{L}, l=1,\dots,L  
\\ & &\forall k=1,\dots,K, \hspace{2mm} \mathrm{if} \hspace{2mm} k \in P_l, \hspace{2mm} \mathrm{then}  \hspace{2mm} f_k(s_1,\dots,s_L)=s_l,
\end{eqnarray}
in which $K'$ is the smallest number larger than or equal to $K$ which is divisible by $L$.

Consider Fig. \ref{Fig_Model_Example}-(a) in which $L=2$ servers are connected to $K=4$ users via a dedicated network. In this example, we have $K'=K$, and it is easy to verify that we can find a routing strategy at intermediate nodes such that:
\begin{eqnarray}
\nonumber
	&&P_1=\{1, 2\}, P_2=\{3, 4\} \\ \nonumber
	&&f_1(s_1,s_2)=f_2(s_1,s_2)=s_1,  \\ \nonumber
	&&f_3(s_1,s_2)=f_4(s_1,s_2)=s_2.
\end{eqnarray}

\item \textbf{Flexible  Networks}

In this class of networks, we assume that there exists a coding (routing) strategy at network nodes such that for \textit{every partitioning} $\{P_1,\dots,P_L\}$ of $[K]=\{1,2, \dots,K\}$ we have:
\begin{eqnarray}
\forall k=1,\dots,K, \hspace{2mm} \mathrm{if} \hspace{2mm} k \in P_l, \hspace{2mm} \mathrm{then}  \hspace{2mm} f_k(s_1,\dots,s_L)=s_l.
\end{eqnarray}

It should be noted that in the dedicated networks, each server was assigned to a \emph{fixed} subset of users, while in flexible networks we can flexibly change these assignments during the data delivery phase. In the example shown in Fig. \ref{Fig_Model_Example}-(b), we have chosen two sample partitionings, i.e. $P_1=\{1,4\}, P_2=\{2,3\}$ for the top figure, and $P_1=\{2,4\}, P_2=\{1,3\}$ for the bottom figure. It is obvious that every flexible network is a dedicated network, but the converse is not true. Hence, flexible networks are generally  \emph{richer} than dedicated networks in terms of their internal connectivity.

\item \textbf{Linear Networks}
	
In the aforementioned dedicated and flexible networks, the intermediate nodes should know the topology of the network in order to do a proper routing of their input data onto their output ports. However, in the case of linear networks, we assume that such knowledge is not available at intermediate nodes. Thus, we assume that each node generates a random linear combination of data at its input ports to be transmitted on its output ports. Consequently, the overall transmit and receive vectors of the network are linearly related at each time slot:
\begin{equation}
\mathbf{r}(t) = \mathbf{H} \mathbf{s}(t),
\end{equation}
where $\mathbf{H} \in  \mathbb{F}_{q}^{K \times L}$. $\mathbf{H}$ is called the \textit{Network Transfer Matrix} (NTM). Let us define:
\begin{eqnarray}
\nonumber \mathbf{X}&\triangleq&[\mathbf{s}(1), \mathbf{s}(2), \dots, \mathbf{s}(T_C)],  
\\ \mathbf{Y}&\triangleq&[\mathbf{r}(1), \mathbf{r}(2), \dots, \mathbf{r}(T_C)].
\end{eqnarray}
We call matrices $\mathbf{X} \in \mathbb{F}_{q}^{L \times T_C}$ and $\mathbf{Y} \in \mathbb{F}_{q}^{K \times T_C}$, transmit and receive blocks, respectively. Then, transmit and receive blocks are also linearly related:
\begin{equation}
	\mathbf{Y}=\mathbf{H}\mathbf{X}.
\end{equation}

In \emph{Linear Networks}, we assume that network topology is \emph{rich-enough} to guarantee that the elements of $\mathbf{H}$ are i.i.d. random variables. Similar to most existing papers employing random linear network coding, we assume large-enough $q=2^m$ to assure that NTM exhibits full rank matrix properties, with high probability \cite{Yang_2011, Li_2003}. Also, we assume uniform distribution on the elements of $\mathbf{H}$, which is a proper assumption for large scale networks with many sources of randomness \cite{Mahdi_2008, Mahdi_2011, Silva_2010, Yang_2010}. 

Finally, for later reference, define $\mathbf{h}_k$ as
\begin{eqnarray}\label{Eq_Model_Linear_Channel_Vector_Def}
\mathbf{h}_k \triangleq [h_{k,1},\dots,h_{k,L}]^t, k=1,\dots,K.
\end{eqnarray}

It should be noted that we assume a static network transfer matrix $\mathbf{H}$, such that it does not change for the duration of $T_C$ time slots. As changes in the network transfer matrix is due to topology changes (e.g. failure of a node), such assumption is valid in most practical scenarios. Fig. \ref{Fig_Model_Example}-(c) illustrates an example of a linear network in the case of $L=2$ and $K=4$. 

\end{itemize}

Finally, it should be noted that in this paper, we assume $N \geq K$. Such assumption will lead to more clear presentation in the rest of this paper and will also exclude the possibility of using uncoded multi-casting schemes that may trivially be adopted for the case of small number of files. Extending the results to the case $N <K$ is straightforward, and the readers are referred to \cite{Maddah-Ali_Fundamental_2014}.

\begin{myremark}
	It should be noted that if a server is connected to the network by a number of links (each of integer capacity) with the total capacity of $t$ symbols per time slot, our model can accommodate this scenario by splitting this server into $t$ separate servers.
\end{myremark}

\begin{myremark}
The random linear network coding approach at intermediate nodes is also used in other papers such as \cite{Das_2010}, \cite{Meng_2012}, and \cite{Meng_2013}, in the context of uni-casting via interference alignment.
\end{myremark}

\section{Main Results: Review and Discussion}\label{Sec_Main_Results}

\begin{figure}
\begin{center}
\includegraphics[width=0.4\textwidth]{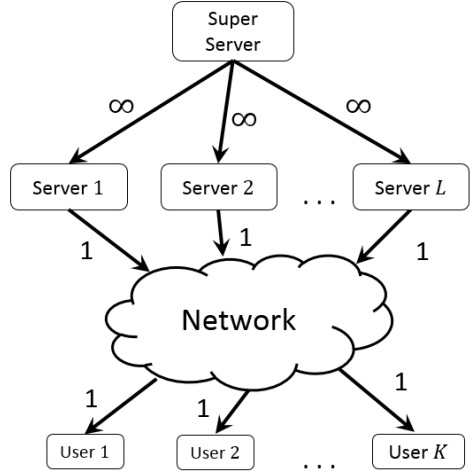}
\end{center}
\caption{The Super-Server Strategy.\label{Fig_Main_Result_1}}
\end{figure}

The simplest approach in designing a coding scheme for the multi-server case is to directly transform it to a single-server scenario and use the scheme presented in \cite{Maddah-Ali_Fundamental_2014}. Such approach can be simply adopted by adding a \emph{Super Server} node and connecting it with edges of infinite capacity to all other servers (see Fig. \ref{Fig_Main_Result_1}). As shown in \cite{Maddah-Ali_Decentralized_2014}, we only need to route packets that are transmitted by the super-server to those users that can benefit from receiving them. For tree networks, such approach results in the following simple topology-aware routing scheme: at each interior node, the packets received at the input port benefiting at least one of the descendants of the node, is sent on the corresponding  output port. As proved in \cite{Maddah-Ali_Decentralized_2014}, the minimum traffic load imposed on each link, in the scaling sense, can be achieved by such simple routing scheme. Such approach also leads to an order-optimal coding delay for tree networks under our formulation.   

One can, however, think of another naive and simple approach to the multi-server problem. We can simply dedicate each server to a subset of users and make it responsible for satisfying the requests of the corresponding subset of users. It is clear that, in order to prevent congestion at a specific server, we should balance out loads of the servers so that each of $L$ servers will be responsible for about $K/L$ users. Consequently, one can easily arrive at the following theorem for the coding delay in dedicated networks:

\begin{thm}\label{Thm_dedicated}
The coding delay for a dedicated network is upper bounded by a piecewise-linear curve with corner points
\begin{equation}\label{Eq_Dedicated_Theorem}
	D^*(M) \leq \frac{K' \left( 1-\frac{M}{N} \right)}{\min{\left(K',L+K'\frac{M}{N}\right)}} \frac{F}{m},
\end{equation}
where $\frac{K'M}{LN} \in \mathbb{N}$ should be satisfied, and $K'$ is the smallest number larger than or equal to $K$ which is divisible by $L$.
\end{thm}  

The proof of Theorem \ref{Thm_dedicated} is straightforward, and thus, we just draw the main sketch here. First, let us review the main concept behind the coded caching scheme for a single server in a broadcast scenario \cite{Maddah-Ali_Fundamental_2014}. In this case, if we do not have any cache at the users, it is clear that the server should in sequence send all the requested files to the users (considering that the users request different files). This will lead to a total amount of $KF$ bits  to be transmitted. Since the server is only able to transmit $m$ bits (a symbol in $\mathbb{F}_q$) at each time slot, the coding delay will be $K\frac{F}{m}$ time slots. By providing cache at the users, the \emph{local caching gain} will reduce the coding delay to $K(1-\frac{M}{N})\frac{F}{m}$. The main result in \cite{Maddah-Ali_Fundamental_2014} indicates that by exploiting the additional \emph{global caching gain}, the coding delay for $KM/N \in \mathbb{N}$ reduces to:
\begin{equation}\label{Eq_Single_Server_Coding_Delay}
T_C=\frac{K(1-M/N)}{1+KM/N}\frac{F}{m},
\end{equation}
which is order optimal for this scenario.

As we extend to the multi-server case, let us assume for simplicity that $K$ is divisible by $L$. Splitting the original $L$-server problem with $K$ users into $L$ single-server problems with $\frac{K}{L}$ users is possible in this case. Since the sub-networks may operate in parallel, the delay is further reduced to:
\begin{eqnarray}
\nonumber	T_C&=&\frac{\frac{K}{L}\left(1-\frac{M}{N}\right)}{1+\frac{K}{L}\frac{M}{N}} \frac{F}{m}\\ \nonumber
&=& \frac{K(1-M/N)}{L+KM/N}\frac{F}{m},
\end{eqnarray}
where $KM/LN \in \mathbb{N}$. Since in any scheme we can benefit at most all the $K$ users simultaneously, the total multi-casting gain of any scheme is at most $K$, and the denominator should be compared to $K$ (by the $\min$ operator in the denominator of (\ref{Eq_Dedicated_Theorem})). Extension to the case where $K$ is not divisible by $L$ can be accomplished by adding virtual users. The following example compares the above two naive approaches:

\begin{examp}\label{Examp_Intro_1}

Consider the network shown in Fig. \ref{Fig_Intro_Examp1} for $K=4$ users. We also assume the library contains $N=4$ files, and each user can store $M=2$ files during the cache content placement phase. By adding a super server a tree network is obtained, and in the delivery phase, the scheme in \cite{Maddah-Ali_Decentralized_2014} suggests to send 
\begin{eqnarray}
\nonumber	R_1&=&\frac{K(1-M/N)}{1+KM/N}F \\ \nonumber
&=&\frac{4\left(1-\frac{2}{4}\right)}{1+\frac{4\times 2}{4}} F \\ \nonumber
&=&\frac{2}{3} F,
\end{eqnarray}
bits at the super server's output. In their scheme, at each node only those packets benefiting the descendants of an output port will be copied on that port.  However, in our case each packet benefits $1+\frac{KM}{N}=3$ users, and thus should be copied on both output ports of node $n_1$. This results in:
\begin{eqnarray}
\nonumber 	R_2=R_1, 
\end{eqnarray}
and since we assumed a capacity of one symbol per time slot for each internal edge, the delay of this scheme is:
\begin{eqnarray}
	T_C=\frac{R_2}{m} = \frac{2}{3} \frac{F}{m}.
\end{eqnarray}

At this stage, the key question is whether it is possible to further reduce the required number of time slots or not? In fact, with a closer look at this network it becomes evident that we can reduce this network to a dedicated network with:
\begin{eqnarray}
\nonumber	P_1&=&\{1,2\} \\ \nonumber
	P_2&=&\{3,4\}.
\end{eqnarray}
Therefore, the original problem can be divided into two sub-problems (see Fig. \ref{Fig_Intro_Examp1}) and each server can address the load of its corresponding sub-network by:
\begin{eqnarray}
\nonumber	R_3&=&\frac{\frac{K}{L}(1-M/N)}{1+\frac{K}{L}\frac{M}{N}}F \\ \nonumber
	&=& \frac{F}{2}.
\end{eqnarray}
Since the sub-networks operate in parallel, the delay of this scheme will be
\begin{eqnarray}
	T_C=\frac{R_3}{m}=\frac{1}{2}\frac{F}{m}
\end{eqnarray}
time slots.
\begin{figure}
\begin{center}
\includegraphics[width=0.8\textwidth]{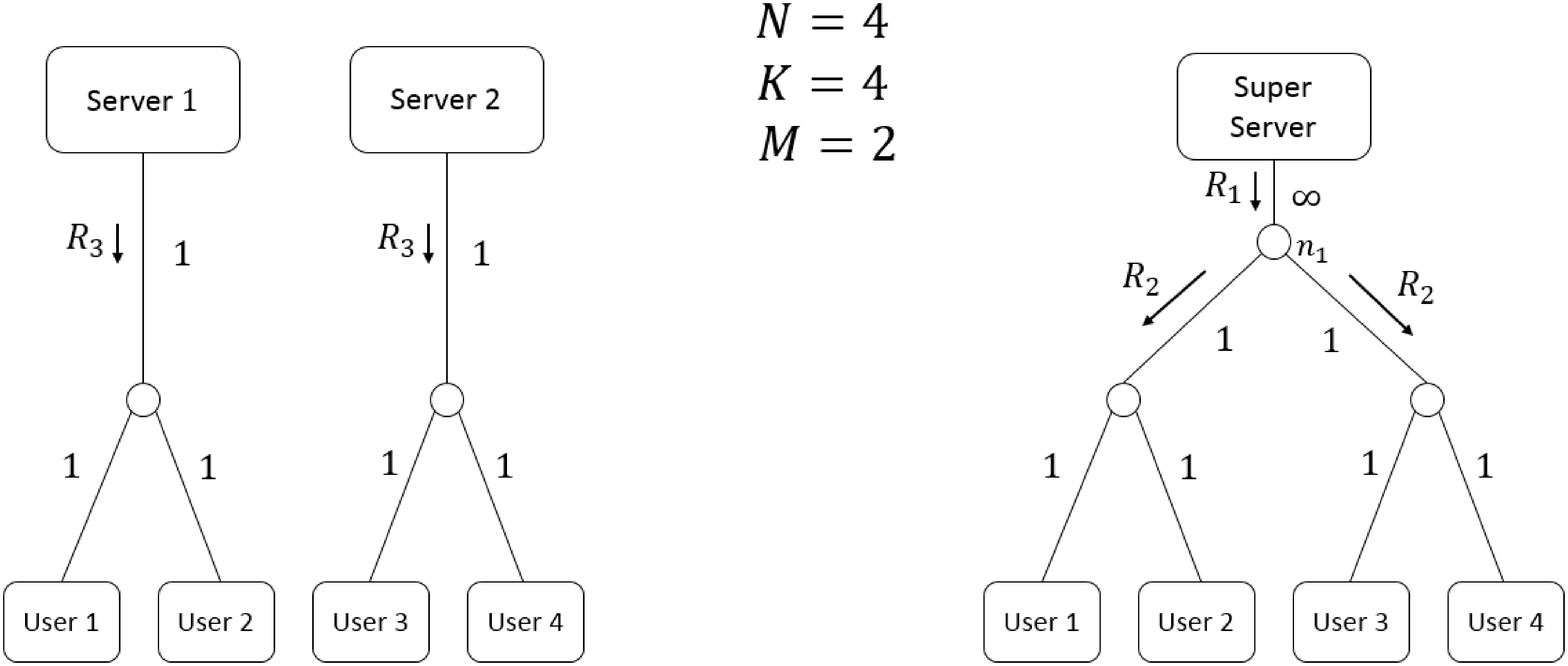}
\end{center}
\caption{Example \ref{Examp_Intro_1}. \label{Fig_Intro_Examp1}}
\end{figure}
\end{examp}

The above example shows that although the scheme in \cite{Maddah-Ali_Decentralized_2014} is order-optimal for tree networks, however, by designing a topology-aware scheme  it may be possible to  arrive at a better pre-constant factor.

Next, let us consider another class of networks with more flexibility, i.e. Flexible Networks. In such networks, similar to dedicated networks, we can assign a subset of users to each server, and the network allows parallel operation of the servers. However, unlike dedicated networks, such assignment can be changed arbitrarily in subsequent transmissions. Such extra freedom in user assignments allows a significant reduction in the coding delay as shown in the following example. 

\begin{examp}[$L=2,K=4,N=4,M=1$]\label{Examp_Flexible_1}
For a single server case, the scheme proposed in \cite{Maddah-Ali_Fundamental_2014} achieves the following delay for $M=1$:
\begin{eqnarray}
\nonumber T_C=\frac{K(1-M/N)}{1+KM/N}\frac{F}{m}=\frac{3}{2} \frac{F}{m}.
\end{eqnarray}
In  order to get a better insight on this result, consider Fig. \ref{Fig_Flexible_Exmp1}-(a) which shows the cache content placement and the delivery scheme for requests $A,B,C,D$ by users $1,2,3,4$, respectively. In the cache content placement phase, each file is divided into four equal-sized parts and cached as shown in Fig. \ref{Fig_Flexible_Exmp1}-(a). In the delivery phase, the single server sends the following data in sequence:
\begin{eqnarray}
\nonumber A_2+B_1, A_3+C_1, A_4+D_1, B_3+C_2, B_4+D_2, C_4+D_3.
\end{eqnarray} 
As a result, six transmissions are required while each has the delay $\frac{1}{4}\frac{F}{m}$. Thus, the total delay will be $T_C=\frac{6}{4}\frac{F}{m}=\frac{3}{2}\frac{F}{m}$. In the above scheme, each transmission benefits a pair of users, and is of no value for the other pair. 

If we have two servers, by the definition of flexible networks each server is able to transmit a given data to a pair of users simultaneously and interference-free from transmission of the other server. In Fig. \ref{Fig_Flexible_Exmp1}-(b), transmissions of the left and right servers are colored as blue and red, respectively.  Thus, a pair of transmissions  in Fig.  \ref{Fig_Flexible_Exmp1}-(a) can be sent simultaneously as shown in Fig.  \ref{Fig_Flexible_Exmp1}-(b), resulting in the achievable pair $(M,T_C)=(1,\frac{3}{4}\frac{F}{m})$. Thus, exploiting the extra flexibility of the network in this example results in the coding delay enhancement, compared with the single-server case.

\end{examp}

\begin{figure}
\begin{center}
\includegraphics[width=0.8\textwidth]{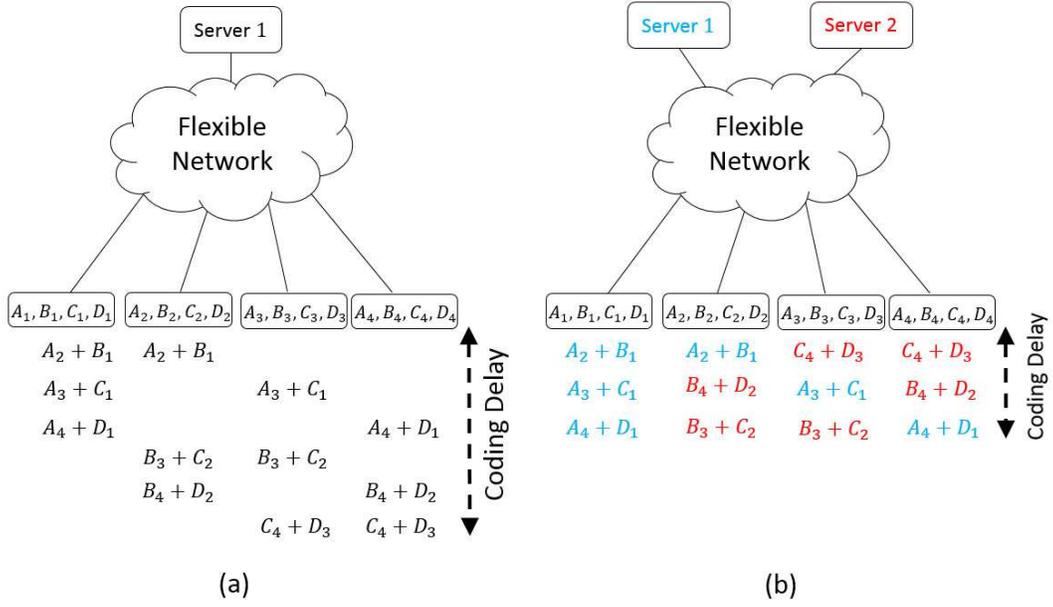}
\end{center}
\caption{Flexible Network Example \ref{Examp_Flexible_1}. \label{Fig_Flexible_Exmp1}}
\end{figure}

In dedicated networks, we exploit the network topology to assign a fixed number of users to each server. In this way, a user receives packets only from a certain server and this assignment is fixed during the course of transmission. In flexible networks, however, at different time slots users can be served by different servers where the assignment strategy is fixed for each server. Fig. \ref{Fig_Strategy} shows two servers connected to three users through such flexible network. The blue packets originating from server 1 are intended for one user (which may change at different time slots) and the red packets originating from server 2 are intended for two users (which may change at different time slots). We assign blue packets to be associated with Strategy $1$ and red packets with Strategy $2$. Fig. \ref{Fig_Strategy} shows consequent transmissions in such network where Strategy 1 is associated with server 1 and Strategy 2 with server 2. In general, we associate Strategy $p$ to a packet if it is intended for $p$ users. Now, if we fix a strategy for a server, it means that all the packets transmitted by that server have the same strategy. It is worth mentioning that packets received by a user do not necessarily have the same strategy, since they may have arrived from different servers (see Fig. \ref{Fig_Strategy}).  

\begin{figure}
\begin{center}
\includegraphics[width=0.8\textwidth]{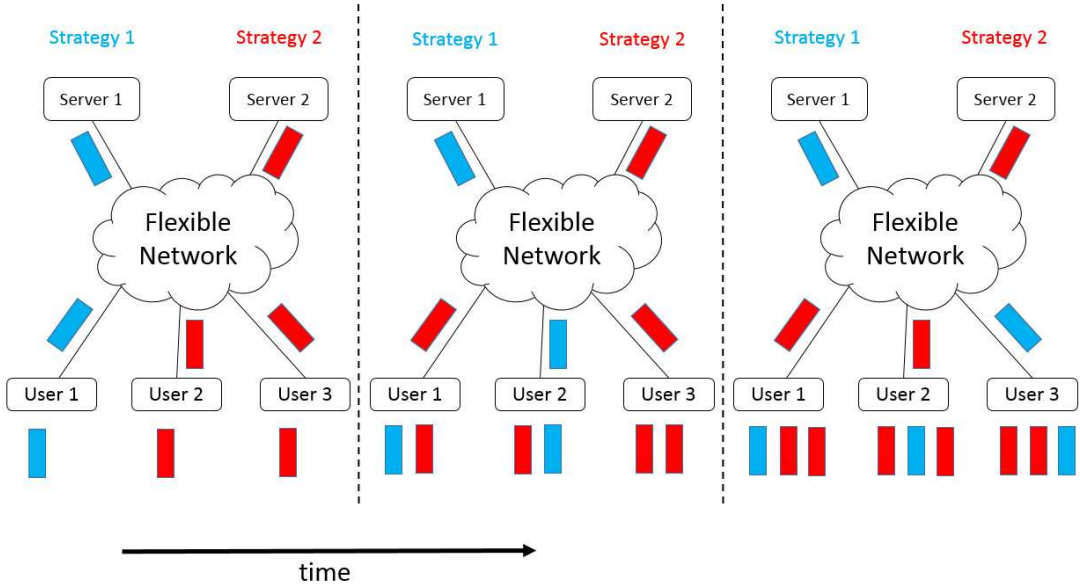}
\end{center}
\caption{Server 1 and blue packets are associated with Strategy 1, and Server 2 and red packets are associated with Strategy 2. \label{Fig_Strategy}}
\end{figure}

Consider server $i$ with Strategy $p_i$. Also, we assign a fraction $F_i$ bits of each file to be delivered by Server $i$. In order to employ the scheme in \cite{Maddah-Ali_Fundamental_2014} for this server, we allocate  a memory of size $\bar{M_i}$ bits from all the users to be used only by Server $i$ where    
\begin{eqnarray}
\nonumber \bar{M_i}=\frac{N}{K}\left(p_i-1\right)F_i.
\end{eqnarray}
Therefore, Server $i$ can deliver $F_i$ bits to all the users in $T_C(i)$ time slots where
\begin{eqnarray}
\nonumber T_C(i)&=&\frac{K\left(1-\frac{\bar{M_i}/F_i}{N}\right)}{1+\frac{K\bar{M_i}/F_i}{N}}\frac{F_i}{m} \\
&=&\frac{K-p_i+1}{p_i}\frac{F_i}{m}.
\end{eqnarray}
We assume that a routing strategy exists where packets from different servers do not interfere with each other. In this case, the total delay is limited by the maximum delay of the servers. Therefore, in order to balance out the servers' loads, we can simply set:
\begin{eqnarray}
\nonumber F_i=\alpha \frac{p_i}{K-p_i+1}F,
\end{eqnarray}
where $\alpha$ does not depend on $i$ and satisfies:
\begin{eqnarray}
\nonumber 	\sum_{i=1}^{L}{F_i}=\alpha \sum_{i=1}^{L}{\frac{p_i}{K-p_i+1}}F=F.
\end{eqnarray}
Therefore, 
\begin{eqnarray}
	\alpha=1/\sum_{i=1}^{L}{\frac{p_i}{K-p_i+1}}.
\end{eqnarray}
Since the total memory is $M$, we have
\begin{eqnarray}
\nonumber	M=\sum_{i=1}^{L} {\bar{M_i}} /F &=& \frac{N}{KF} \sum_{i=1}^{L}{\left(p_i-1\right)F_i} \\ \nonumber &=& \frac{N}{K} \sum_{i=1}^{L}{\left(p_i-1\right)\alpha \frac{p_i}{K-p_i+1}} \\ 
&=& \frac{N}{K} \frac{\sum_{i=1}^{L}{\frac{p_i(p_i-1)}{K-p_i+1}}}{\sum_{i=1}^{L}{\frac{p_i}{K-p_i+1}}}.
\end{eqnarray}
Hence,
\begin{eqnarray}
\nonumber 	T_C&=&\alpha \frac{F}{m} \\ 
		&=& \frac{\frac{F}{m}}{\sum_{i=1}^{L}{\frac{p_i}{K-p_i+1}}}.
\end{eqnarray}

The aforementioned result is based on a strong assumption that a routing strategy exists for parallel and interference-free transmission of the packets. In Section \ref{Sec_Flexible}, we show that such a strategy does in fact exist for flexible networks. The preceding discussion is a rough proof of the following Theorem:

\begin{thm}\label{Thm_flexible}
Suppose a flexible network with $L$ servers. Then, for all $Q \in \{0,\dots,K-L\}$ the following $(M,T_C)$ pairs (and the straight lines connecting them) are achievable
\begin{equation}\label{Eq_Main_Resuls_Flexile_Th}
(M,T_C)=\left\{\left(\frac{N}{K}\frac{\sum_1^L{\frac{p_i(p_i-1)}{K-p_i+1}}}{\sum_1^L{\frac{p_i}{K-p_i+1}}},\frac{1}{\sum_1^L{\frac{p_i}{K-p_i+1}}}\frac{F}{m}\right) \mathrm{,\hspace{1 mm} for \hspace{1 mm} all \hspace{5 mm}} p_1+\dots+p_{L}=K-Q \mathrm{,\hspace{1 mm} where \hspace{5 mm}} p_i \geq 2 \right\},
\end{equation}
and thus lead to an upper bound for the optimum coding delay $D^*$. 
\end{thm}

\begin{proof}
See Section \ref{Sec_Flexible} for the proof.
\end{proof}

In the following example, we present a network in which employing the flexible network strategy results will go beyond earlier results and paves the way for  scaling improvement in the coding delay compared with the super-server strategy.

\begin{examp}\label{Examp_Intro_2}
Consider the network depicted in Fig. \ref{Fig_Intro_Examp2}-(a). In this network, $L$ (an even number) servers  are  connected to $K=L^2/2$  users via $L$ intermediate nodes where each intermediate node has dedicated links to all the users. We also assume:
\begin{eqnarray}
\nonumber \frac{M}{N}=\frac{2}{L^2}\left(\frac{L}{2}-1\right).
\end{eqnarray}
In order to use the super-server strategy with the tree approach proposed in \cite{Maddah-Ali_Decentralized_2014}, we need to choose an appropriate tree inside the network. It can be easily verified that the tree illustrated in Fig. \ref{Fig_Intro_Examp2}-(b) is the best choice. Therefore, $R_1$, the minimum rate of the super-server, is given by 
\begin{eqnarray}
\nonumber	R_1&=&\frac{K(1-M/N)}{1+KM/N} F \\ \nonumber
	&=&\frac{\frac{L^2}{2}(1-M/N)}{\frac{L}{2}} F \\ \nonumber
	&=& L(1-M/N) F.
\end{eqnarray}
The load $R_2$ on each server consists of those packets that are useful for at least a user which is a descendant of that server. We know that each packet benefits a subset of users of size:
\begin{eqnarray}
\nonumber	1+\frac{KM}{N}=\frac{L}{2}.
\end{eqnarray} 
Therefore, the ratio of packets routed on a specific edge to the total number of packets is:
\begin{eqnarray}
\nonumber	\frac{R_2}{R_1}&=& \frac{\sum_{i=1}^{L/2}{{L/2 \choose i}{{L^2/2-L/2 \choose L/2-i}}}}{{L^2/2 \choose L/2}}     \\ \nonumber
&=& 
1-\left(1-\frac{L/2}{L^2/2}\right)  \left(1-\frac{L/2}{L^2/2-1}\right) \dots \left(1-\frac{L/2}{L^2/2-(L/2-1)} 
\right)  \\ \nonumber
&\geq&  1- \left(1-\frac{1}{L}\right)^{\frac{L}{2}} \\ \nonumber
&\sim&1-e^{-1/2},
\end{eqnarray}
for large $L$. Thus, almost a constant number of packets generated by the server will be routed on each edge. This will result in a delay of:

\begin{eqnarray}\label{Eq_MainResults_Examp3_Delay_Super}
\nonumber	T_C&=& \frac{R_2}{m} \\ 
		&\sim& \left(1-e^{-1/2}\right) L(1-\frac{M}{N}) \frac{F}{m}
\end{eqnarray}	
time slots. 

A closer look at the network topology shows that the network is indeed flexible. Setting $p_i=L/2$ which satisfies $\sum{p_i}=K$ and using memory size $M$ where
\begin{eqnarray}	\nonumber M&=&\frac{N}{K}\frac{\sum_1^L{\frac{p_i(p_i-1)}{K-p_i+1}}}{\sum_1^L{\frac{p_i}{K-p_i+1}}} \\ \nonumber
&=& \frac{N}{K} \left(\frac{L}{2}-1\right) \\ 
&=& \frac{N}{L^2/2} \left(\frac{L}{2}-1\right),
\end{eqnarray}
Theorem \ref{Thm_flexible} can be used to achieve the following coding delay:
\begin{eqnarray}\label{Eq_MainResults_Examp3_Delay_Flex}
\nonumber T_C&=&\frac{1}{\sum_1^L{\frac{p_i}{K-p_i+1}}}\frac{F}{m} \\ \nonumber
&=& \frac{F/m}{L \frac{L/2}{L^2/2-L/2+1}} \\ 
&=&\left(1-\frac{M}{N}\right) \frac{F}{m}.
\end{eqnarray}

The above delay in (\ref{Eq_MainResults_Examp3_Delay_Flex}) is not only a scaling improvement compared with the super-server tree-based strategy with delay (\ref{Eq_MainResults_Examp3_Delay_Super}), but also the optimal delay. This is due to the fact that each user can store at most $\frac{M}{N}F$ bits of each file.

\begin{figure}
\begin{center}
\includegraphics[width=0.8\textwidth]{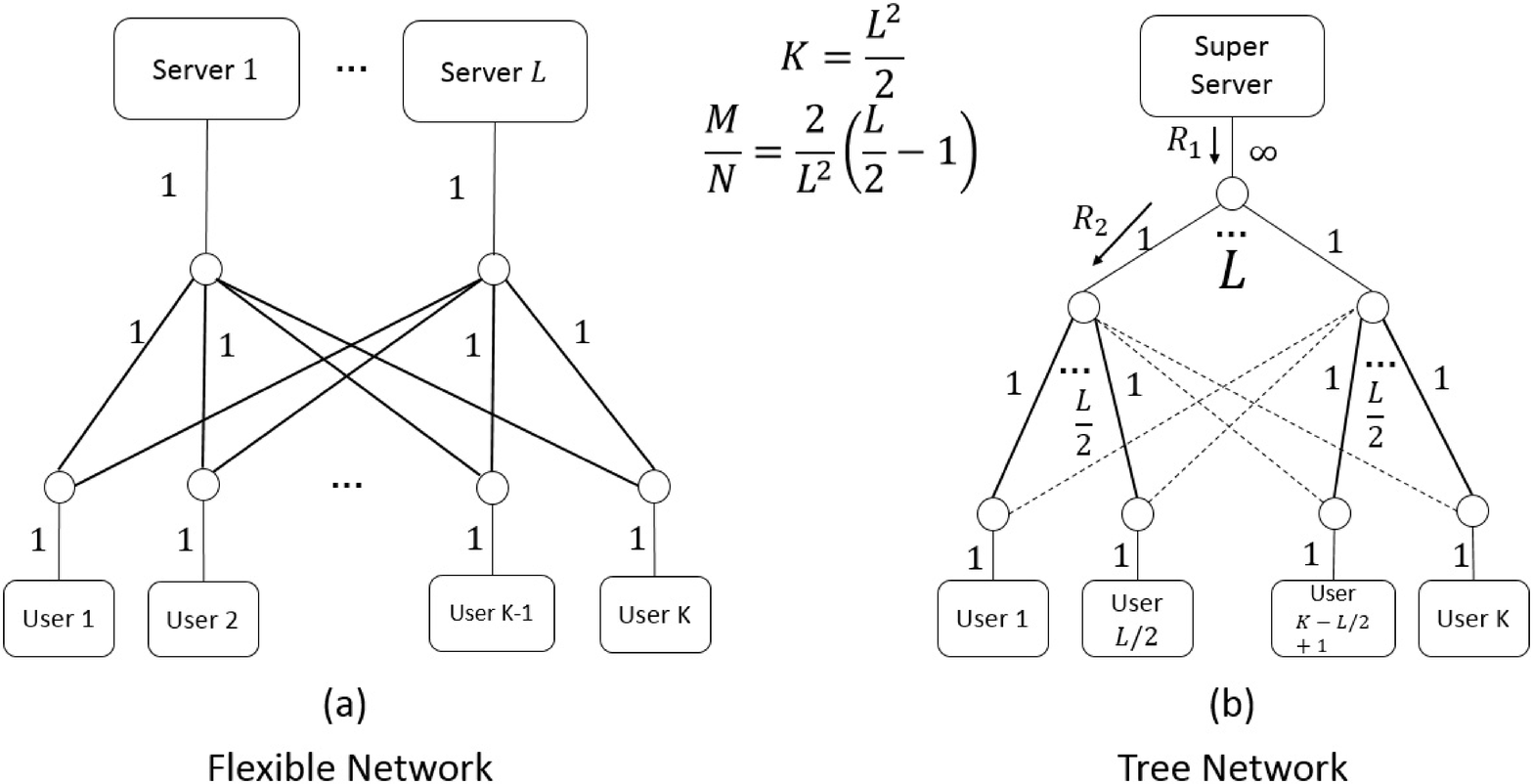}
\end{center}
\caption{Example \ref{Examp_Intro_2}. \label{Fig_Intro_Examp2}}
\end{figure}

\end{examp}

The optimality of the preceding coding scheme can be generalized to any flexible network where $K$ is divisible by $L$ as the following theorem states.
\begin{thm}\label{Thm_Order_Optimal}
	If $K$ is divisible by $L$, then the upper bound in Theorem \ref{Thm_flexible} is optimal within a multiplicative constant gap.
\end{thm}

\begin{proof}
See Section \ref{Sec_Flexible} for the proof.
\end{proof}

For flexible and topologically complex networks, finding a proper routing strategy that achieves the optimal coding delay may not be straightforward. To overcome this difficulty, internal nodes can perform simple random linear network coding which is oblivious to the network's topology. Although this strategy may not be optimal, it has the advantage of being practical and robust. In this way, the network model reduces to a linear network model and the following theorem provides an achievable coding delay for such networks. 

\begin{thm}\label{Thm_linear}
The coding delay for a linear network with $L$ servers is upper bounded by a piecewise-linear curve with the corner points
\begin{equation}
D^*(M) \leq \frac{K(1-M/N)}{\min(K,L+KM/N)}\frac{F}{m},
\end{equation}
where  $KM/N \in \mathbb{N}$ should be satisfied. 
\end{thm}
\begin{proof}
	See Section \ref{Sec_Linear} for the proof.
\end{proof}

In linear networks, a packet intended for a certain number of users, in general, interferes with all other users. Proper pre-coding schemes can be adopted to reduce interference in such networks. Consequently, simultaneous transmission of multiple packets will further reduce network coding delay.  
In order to clarify the implications of Theorem \ref{Thm_linear}, we present the following example:

\begin{figure}
\begin{center}
\includegraphics[width=0.4\textwidth]{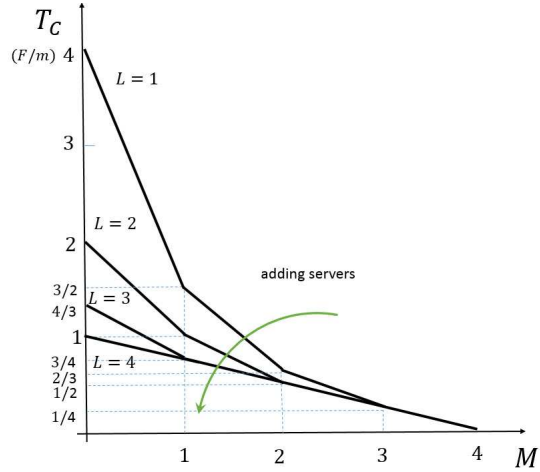}
\end{center}
\caption{Example \ref{Examp_Linear_2}: $N=4, K=4$.\label{Fig_4_4_Curves_Together}}
\end{figure}

\begin{examp}[$K=4, N=4$]\label{Examp_Linear_2}
Consider a network with $K=N=4$. Using  Theorem \ref{Thm_linear}, the coding delay for any $L \in \{1,2,3,4\}$ is given by 
$$T_C = \frac{4-M}{\min(4,L+M)}\frac{F}{m}. $$
The above delay is plotted in Fig. \ref{Fig_4_4_Curves_Together} for $L \in \{1,2,3,4\}$. The problem for $L=1$  reduces to that of \cite{Maddah-Ali_Fundamental_2014}. For $L=4$, we obtain a  multiplexing gain of $4$ by constructing four parallel interference-free links each from one server to one user (e.g. through Singular Value Decomposition) and the optimal coding delay is achieved. Networks with $L\in \{2,3\}$ are interesting cases where interference management is required to achieve the gain $\min(4,L+M)$ in the denominator. The detail of the coding strategy is rather involved and we delegate it to Appendices B and C. 

\end{examp} 

\section{Flexible Networks: Details}\label{Sec_Flexible}

In this section, we present an achievable scheme for the flexible networks leading to the result given in Theorem \ref{Thm_flexible}. We also provide a proof for the optimality result in Theorem \ref{Thm_Order_Optimal} through cut-set analysis.

For the achievability part, we need to provide the cache content placement and content delivery strategies.  Let us start with defining the following parameters: let $Q \in \{0,\dots,K-L\}$ and consider an integer solution of the following equation:
\begin{eqnarray}
\nonumber	p_1+\dots+p_L+p_{L+1}=K,
\end{eqnarray}
where $p_{L+1}=Q$ and $p_i\geq 2, i=1,\dots,L$. We also define 
\begin{eqnarray}
\nonumber	\alpha_i &\triangleq& {K \choose p_i-1}, i=1,\dots,L \\ \nonumber
	\gamma_i &\triangleq& \frac{(K-p_i)!p_i!}{p_1!\dots p_{L+1}!}, i=1,\dots,L+1 \\ \nonumber
	x &\triangleq& 1/\sum_1^L \alpha_i \gamma_i \\  
	x_i &\triangleq& \cases{
   \gamma_i x      &\quad $i=1,\dots,L$ \cr
   0 &\quad $i=L+1$ \cr
 }.
\end{eqnarray}
\\
\textbf{Cache Placement Strategy:} First, split each file $W_n$ into $L$ sub-files
\begin{eqnarray}
\nonumber 	W_n=\left(W_n^i : i=1,\dots,L\right),
\end{eqnarray}
where $W_n^i$ is of size $\alpha_i x_i F$. Then, split each sub-file $W_n^i$ into $\alpha_i$ equal-sized mini-files:
\begin{eqnarray}
\nonumber W_n^i=\left(W_{n,\tau_i}^i: \tau_i \subseteq [K], |\tau_i|=p_i-1 \right).
\end{eqnarray}
Finally, split each mini-file $W_{n,\tau_i}^i$ into $\gamma_i$ equal-sized pico-files of size $xF$ bits:
\begin{eqnarray}
\nonumber  W_{n,\tau_i}^i=\left( W_{n,\tau_i}^{i,j}: j=1,\dots,\gamma_i \right),
\end{eqnarray}
where $\gamma_i$ is an integer number. For each user $k$, we cache pico-file $W_{n,\tau_i}^{i,j}$ if $k \in \tau_i$, for all possible $i,j,n$. Then, the required memory size for each user is:
\begin{eqnarray}
\nonumber  M&=&\frac{1}{F} N \left(  \sum_{i=1}^L{{K-1 \choose p_i-2} \gamma_i x F} \right) \\ \nonumber
	&=& N \frac{\sum_{i=1}^L{ {K-1 \choose p_i-2} \gamma_i }}{\sum_{i=1}^L{ {K \choose p_i-1} \gamma_i }} \\ 
	&=& \frac{N}{K} \frac{ \sum_{i=1}^L {\frac{p_i(p_i-1)}{K-p_i+1}} } { \sum_{i=1}^L {\frac{p_i}{K-p_i+1}} },
\end{eqnarray}
which is consistent with the assumptions of Theorem \ref{Thm_flexible}.
\\
\\
\textbf{Content Delivery Strategy:} Define $P_1^i,\dots,P_{{K \choose p_i}}^i$ to be the collection of all $p_i$-subsets of $[K]$ for all $i=1,\dots L+1$. The delivery phase consists of $\frac{K!}{p_1!\dots p_{L+1}!}$ transmit slots. Each transmit slot is in one-to-one correspondence with one $(p_1,\dots,p_{L+1})$-partition of $[K]$. Consider the transmit slot associated with the partition
\begin{eqnarray}
\nonumber	\left\{ P^1_{\theta_1},\dots,P^{L+1}_{\theta_{L+1}} \right\},
\end{eqnarray}
where $\theta_i \in \left\{1,\dots,{K \choose p_i}\right\}$. Then, the server $i$ sends
\begin{eqnarray}
\nonumber	{+}_{r \in P^i_{\theta_i}}W_{d_r,P^i_{\theta_i} \backslash \{r\}}^{i,N(P^i_{\theta_i})}
\end{eqnarray}
to the subset of users $P^i_{\theta_i}$, interference-free from other servers, where the sum is in $\mathbb{F}_q$ and is over all $r \in P^i_{\theta_i}$. Since we have assumed a flexible network, simultaneous transmissions by all servers is feasible. Also, the index $N(P^i_{\theta_i})$ is chosen such that each new transmission consists of a fresh (not transmitted earlier) pico-file. Obviously, the virtual server $L+1$ does not transmit any packet.

Since each pico-file consists of $x\frac{F}{m}$ symbols, at each transmission slot we should send a block of size $L$-by-$x\frac{F}{m}$ by the servers. Also, since this action should be performed for all $\frac{K!}{p_1!\dots p_{L+1}!}$ slots, the delay of this scheme will be:
\begin{eqnarray}
\nonumber	T_c&=&\frac{K!}{p_1!\dots p_{L+1}!} \times x\frac{F}{m} \\ 
	&=&\frac{1}{\sum_1^L{\frac{p_i}{K-p_i+1}}} \frac{F}{m},
\end{eqnarray}
as stated in Theorem \ref{Thm_flexible}. Consequently, if we show that through the aforementioned number of transmit slots all users will be able to recover their requested files, the proof is complete.
\\
\textbf{Correctness Proof:} The main theme of this scheme is to divide each file into $L$ sub-files, and to assign each sub-file to a single server. Then, each server's task is to deliver the assigned sub-files to the desired users (see Fig. \ref{Fig_Flexible_Proof}).

Consider server $i$. This server handles sub-files $W_n^i, n \in [N]$ though the following delivery tasks:
\begin{eqnarray}
\nonumber & & W_{d_1}^i  \hspace{2mm} 
\stackrel{\textrm{server i} }{\Longrightarrow}   \hspace{2mm} \textrm{User 1} \\ \nonumber
& & W_{d_2}^i  \hspace{2mm} \stackrel{\textrm{server i} }{\Longrightarrow} \hspace{2mm} \textrm{User 2} \\ \nonumber
& & \vdots \\ \nonumber
& & W_{d_K}^i \hspace{2mm} \stackrel{\textrm{server i} }{\Longrightarrow} \hspace{2mm}  \textrm{User K} 
\end{eqnarray}
The above formulation leads to a single server problem \cite{Maddah-Ali_Fundamental_2014} with files of size $F_i= \alpha_i x_i F$ bits. It can be easily verified that the proposed  cache placement strategy for each sub-file mimics that of \cite{Maddah-Ali_Fundamental_2014} for single-server problems. Therefore, if we demonstrate that this server is able to send a common message of size $x_i\frac{F}{m}$ symbols to all $p_i$-subsets of users, then this server can handle this single-server problem successfully. However, in the above scheduling scheme, the server benefits each $p_i$-subset of the users by a common message of size $x\frac{F}{m}$ symbols (a pico-file size), $\gamma_i$ times. Consequently, the total volume of common message that this server is able to send to each $p_i$-subset is $\gamma_i \cdot x\frac{F}{m}=x_i\frac{F}{m}$ symbols. 

Since by proper scheduling scheme in flexible networks all servers can perform the same task simultaneously, all requested portions of files will be delivered. It should be noted that the portion of each file assigned to the virtual server is $x_{L+1}=0$. Algorithm 1 presents the pseudo-code of the procedure described above.

\begin{figure}
\begin{center}
\includegraphics[width=0.7\textwidth]{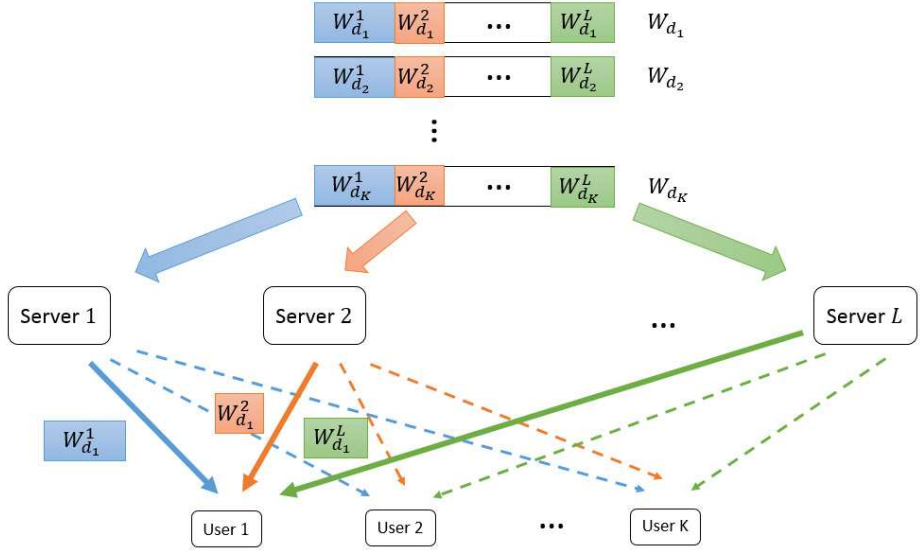}
\end{center}
\caption{Flexible network file distribution for proof of Theorem \ref{Thm_flexible}. \label{Fig_Flexible_Proof}}
\end{figure}

\begin{algorithm}\label{Alg_Main_Flexible}
\caption{Multi-Server Coded Caching - Flexible Networks}
\begin{algorithmic}[1]
\Procedure{PLACEMENT}{$W_1,\dots,W_N,p_1,\dots,p_{L+1}$}
\State $\alpha_i \gets {K \choose p_i-1}$, $i=1,\dots,L$
\State $\gamma_i \gets ((K-p_i)!p_i!)/(p_1!\dots,p_{L+1}!)$, $i=1,\dots,L+1$
\State $x \gets 1/(\sum_1^L{\alpha_i \gamma_i})$
\State $x_i \gets \gamma_i x$, $i=1,\dots,L$
\State $x_{L+1} \gets 0$
\ForAll{$n \in [N]$}
\State split $W_n$ into $(W_n^i: i=1,\dots,L)$, where $|W_n^i|=\alpha_i x_i$
\ForAll{$i=1,\dots,L$}
\State split $W_n^i$ into $(W_{n,\tau_i}^i: \tau_i \subset [K], |\tau_i|=p_i-1)$ of equal size
\ForAll{$\tau_i \subset [K], |\tau_i|=p_i-1$}
\State split $W_{n,\tau_i}^i$ into $(W_{n,\tau_i}^{i,j}: j=1,\dots,\gamma_i)$ of equal size
\EndFor
\EndFor
\EndFor
\ForAll{$k \in [K]$}
\ForAll{$i=1,\dots,L$}
\State $Z_k \gets (W_{n,\tau_i}^{i,j}: \tau_i \subset [K], |\tau_i|=p_i-1, k \in \tau_i, j=1,\dots,\gamma_i, n \in [N])$
\EndFor
\EndFor
\EndProcedure
\\
\Procedure{DELIVERY}{$W_1,\dots,W_N$, $d_1,\dots,d_K$, $p_1,\dots,p_{L+1}$}
\ForAll{$i=1,\dots,L$}
\ForAll{$j=1,\dots,{K \choose p_i}$}
\State $N({P}^i_j) \gets 1 $
\EndFor
\EndFor
\ForAll{partitions of $[K]$ with sizes $p_1,\dots,p_{L+1}$, ($p_i \geq 2, i=1,\dots,L$)}
\State $\{P^1_{\theta_1},\dots,P^{L+1}_{\theta_{L+1}}\} \gets $ selected partition
\State \textbf{transmit} $ \mathbf{X}(\{P^1_{\theta_1},\dots,P^{L+1}_{\theta_{L+1}}\}) = \left[ {\begin{array}{c}
   {+}_{r \in P^1_{\theta_1}}W_{d_r,P^1_{\theta_1} \backslash \{r\}}^{1,N(P^1_{\theta_1})} \Rightarrow P^1_{\theta_1} \\
	\vdots \\
   {+}_{r \in P^L_{\theta_L}} W_{d_r,P^L_{\theta_L} \backslash \{r\}}^{L,N(P^L_{\theta_L})} \Rightarrow P^L_{\theta_L} \\
  \end{array} } \right] $
  
\ForAll{$i=1,\dots,L$}
\State $N(P^i_{\theta_i}) \leftarrow N(P^i_{\theta_i}) +1$
\EndFor   
\EndFor

\EndProcedure

\end{algorithmic}
\end{algorithm}

To prove Theorem \ref{Thm_Order_Optimal}, we first state the following lemma:

\begin{lem}\label{Lem_Converse}
The coding delay for a general network with $L$ servers is lower bounded by 
\begin{equation}\label{Eq_Converse_Theorem}
D^*(M) \geq \max_{s \in \{1,\dots,K\}}\frac{1}{\min(s,L)}\left(s-\frac{s}{\lfloor \frac{N}{s} \rfloor}M\right) \frac{F}{m}.
\end{equation}
\end{lem}
\begin{proof}
See Appendix A for the proof.
\end{proof}

The above lemma can be used to prove optimality of the proposed scheme in some range of parameters. The following corollary states the result. 
\begin{cor}\label{Cor_Converse}
	All $(M-T_C)$ pairs in Theorem \ref{Thm_flexible} corresponding to $Q=0$ are optimal. Thus, the converse line $\left(1-\frac{M}{N}\right)\frac{F}{m}$ is achieved for $M^*\leq M \leq N$, where
\begin{equation}
M^*=\min_{p_1+\dots+p_L=K}\frac{N}{K}\frac{\sum_1^L{\frac{p_i(p_i-1)}{K-p_i+1}}}{\sum_1^L{\frac{p_i}{K-p_i+1}}}.
\end{equation}
\end{cor}
\begin{proof}
Theorem \ref{Thm_flexible} states that all the $(M-T_C)$ pairs in (\ref{Eq_Main_Resuls_Flexile_Th}) are achievable. By some simple calculations one can show that for these achievable pairs we have:
\begin{eqnarray}\label{Eq_Flexible_Converse_Cor_1}
	\left(1-\frac{M}{N}\right)\frac{F}{m}=\left(1-\frac{Q}{K}\right) T_C.
\end{eqnarray}
Therefore, if we put $Q=0$ in Theorem \ref{Thm_flexible}, all the corresponding $(M-T_C)$ pairs satisfy
\begin{eqnarray}
\nonumber	T_C=\left(1-\frac{M}{N}\right)\frac{F}{m}.
\end{eqnarray}
On the other hand, by considering the case of $s=1$ in Lemma \ref{Lem_Converse} we know that the optimal coding delay satisfies:
\begin{eqnarray}
\nonumber	D^*(M) \geq \left(1-\frac{M}{N}\right)\frac{F}{m},
\end{eqnarray} 
which is matched to our achievable coding delay . Therefore, by setting $Q=0$ in  Theorem \ref{Thm_flexible}, for all 
\begin{eqnarray}
\nonumber M=\frac{N}{K}\frac{\sum_1^L{\frac{p_i(p_i-1)}{K-p_i+1}}}{\sum_1^L{\frac{p_i}{K-p_i+1}}}, p_1+\dots,p_L=K, p_i \geq 2,
\end{eqnarray} 
the achievable coding delay is optimum. By minimizing the cache size, over all partitionings satisfying $p_1+\dots,p_L=K, p_i \geq 2$, the proof is complete.
\end{proof}

There is an interesting intuition behind Eq. (\ref{Eq_Flexible_Converse_Cor_1}). In the proposed scheme for flexible networks, we assigned a subset of $Q$ users to the virtual server, and all the other $K-Q$ users benefited from other servers. Thus, through each transmission, the ratio $\frac{K-Q}{K}$ of users will be real users. This is exactly the coefficient that shows how close is the achieved delay to the optimal curve $(1-M/N)F/m$.

Finally, we are ready to prove Theorem \ref{Thm_Order_Optimal}. We consider two regimes for cache sizes. First , we let 
\begin{eqnarray}
\nonumber	M^*=\frac{N}{K}\left(\frac{K}{L}-1\right).
\end{eqnarray}
In the first regime where $M \geq M^*$, using Theorem \ref{Thm_flexible} with $Q=0$ and $p_1,\dots,p_L=\frac{K}{L}$, we obtain:
\begin{eqnarray}
\nonumber T_C = \left(1-\frac{M}{N}\right)\frac{F}{m}.
\end{eqnarray} 
As Corollary \ref{Cor_Converse} states, for this case the optimal curve is achieved. 

For the second regime where  $M<M^*$ (such that $KM/N \in \mathbb{N}$), set
\begin{eqnarray}
\nonumber	& &Q=K-\left(\frac{KM}{N}+1\right)L \\ \nonumber
	& & p_1,\dots,p_L=\frac{K-Q}{L} = \left(\frac{KM}{N}+1\right).
\end{eqnarray}
Then, we obtain:
\begin{eqnarray}
\nonumber	T_C = \frac{1}{L}\frac{K(1-M/N)}{1+KM/N}.
\end{eqnarray}

On the other hand, from Lemma \ref{Lem_Converse} we have:

\begin{eqnarray}
\nonumber D^* &\geq& \max_{s \in \{1,\dots,K\}}\frac{1}{\min(s,L)}\left(s-\frac{s}{\lfloor \frac{N}{s} \rfloor}M\right) \frac{F}{m} \\ \nonumber
&\geq& \max_{s \in \{1,\dots,K\}}\frac{1}{L}\left(s-\frac{s}{\lfloor \frac{N}{s} \rfloor}M\right) \frac{F}{m} \\ \nonumber
&\geq& \frac{1}{L}\frac{1}{12}\frac{K(1-M/N)}{1+KM/N}, \\
&\geq& \frac{1}{12} T_C,
\end{eqnarray}
where the last inequality follows from \cite{Maddah-Ali_Fundamental_2014}. This concludes the proof of Theorem \ref{Thm_Order_Optimal}.

\section{Linear Networks: Details}\label{Sec_Linear}

In order to explain the main concepts behind the coding strategy proposed for linear networks, we will first present a simple example:

\begin{figure}
\begin{center}
\includegraphics[width=0.5\textwidth]{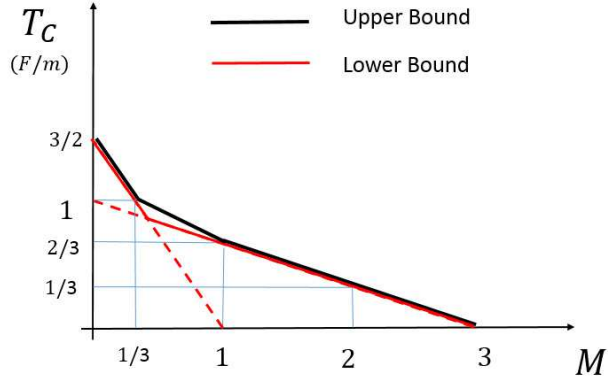}
\end{center}
\caption{Example \ref{Examp_Linear_1} ($L=2, K=3, N=3$): Lower and upper bounds on the coding delay.\label{Fig_Multi_Server_3_3_2_Example}}
\end{figure}

\begin{examp}[$L=2, K=3, N=3$]\label{Examp_Linear_1}
In this example, we consider a network consisting of $L=2$ servers, $K=3$ users, and a library of $N=3$ files, namely $W_1=A$, $W_2=B$, and $W_3=C$. By definition of linear networks the input-output relation of this network is characterized by a $3$-by-$2$ random matrix $\mathbf{H}$. Lower and upper bounds for the coding delay of this setting are shown in Fig. \ref{Fig_Multi_Server_3_3_2_Example}. The lower bound is due to Lemma \ref{Lem_Converse} as follows:
\begin{equation}
D^* \geq \max\left(1-\frac{M}{3},\frac{3-3M}{2} \right)\frac{F}{m}.
\end{equation}
The upper bound is due to Theorem \ref{Thm_linear} except the achievable pair $(M,T_C)=(\frac{1}{3},1)$, which will be discussed later. We have also exploited the fact that the straight line connecting every two achievable points on the $M-T_C$ curve is also achievable, as shown in \cite{Maddah-Ali_Fundamental_2014}. In order to get a glimpse of the ideas of the coding strategy behind Theorem \ref{Thm_linear}, next we discuss the achievable $(M,T_C)$ pair $(1,\frac{2}{3})$. In this case, as we will show, we can benefit both from the local/global caching gain (provided by cache of each user), and the multiplexing gain (provided by multiple servers in the network). The question is how to design an scheme so that we can exploit both gains simultaneously. In what follows we provide the solution:

Suppose that (without loss of generality) in the second phase, the first, second, and third users request files $A$, $B$, and $C$ respectively. Assume that the cache content placement is similar to that of \cite{Maddah-Ali_Fundamental_2014}: First, divide each file into three equal-sized non-overlapping sub-files:
\begin{eqnarray}
\nonumber A&=&[A_1,A_2,A_3] \\ \nonumber
B&=&[B_1,B_2,B_3] \\ \nonumber
C&=&[C_1,C_2,C_3].
\end{eqnarray}
Then, put the following contents in the cache of users:
\begin{eqnarray}
\nonumber Z_1&=&[A_1,B_1,C_1] \\ \nonumber
Z_2&=&[A_2,B_2,C_2] \\ \nonumber
Z_3&=&[A_3,B_3,C_3].
\end{eqnarray}
Let $L(x_1,\dots,x_m)$ be a random linear combination of $x_1,\dots,x_m$ as defined earlier. Consequently, in this strategy, the two servers send the following transmit block:
\begin{equation}\label{Eq_L2_K3_N3_M1_transmit}
\mathbf{X}=[\mathbf{h}_1^{\perp}L_1^1(C_2,B_3)+\mathbf{h}_2^{\perp}L_2^1(A_3,C_1)+\mathbf{h}_3^{\perp}L_3^1(A_2,B_1), \mathbf{h}_1^{\perp}L_1^2(C_2,B_3)+\mathbf{h}_2^{\perp}L_2^2(A_3,C_1)+\mathbf{h}_3^{\perp}L_3^2(A_2,B_1) ].
\end{equation}
where the random linear combination operator $L(\cdot,\cdot)$ operates on sub-files, in an element-wise manner, and $\mathbf{h}_i^{\perp}$ is an orthogonal vector to $\mathbf{h}_i$ (i.e. $\mathbf{h}_i.\mathbf{h}_i^{\perp}=0$). Let us focus on the first user who will receive:
\begin{eqnarray}
\nonumber \mathbf{h}_1.\mathbf{X}&=&[(\mathbf{h}_2^{\perp} . \mathbf{h}_1) L_2^1(A_3,C_1)+(\mathbf{h}_3^{\perp} . \mathbf{h}_1)L_3^1(A_2,B_1), (\mathbf{h}_2^{\perp} . \mathbf{h}_1)L_2^2(A_3,C_1)+(\mathbf{h}_3^{\perp} . \mathbf{h}_1)L_3^2(A_2,B_1) ] \\ 
&=& [L^1(A_2,A_3,C_1,B_1),L^2(A_2,A_3,B_1,C_1)].
\end{eqnarray} 
As the first user has already cached $B_1$ and $C_1$ in the first phase, by subtracting the effect of interference terms, the first user can recover:
\begin{eqnarray}
\nonumber [L(A_2,A_3),L'(A_2,A_3)],
\end{eqnarray}
which consists of two independent (with high probability for large field size $q$) linear combinations  of $A_2$ and $A_3$. By solving these independent linear equations, such user can decode $A_2$ and $A_3$, and with the help of $A_1$ cached at the first phase, he can recover the whole requested file $A$. It can easily be verified that other users can also decode their requested files in a similar fashion. The transmit block size indicated in (\ref{Eq_L2_K3_N3_M1_transmit}) is $2$-by-$\frac{2F}{3m}$, resulting in $T_C=\frac{2F}{3m}$ time slots.

Let us forget about one of the servers for a moment and assume we have just one server. Then, the scheme proposed in \cite{Maddah-Ali_Fundamental_2014} only benefits two users per transmission through pure global caching gain. Also, in the case of two servers and no cache memory (the aforementioned case of $M=0$), we could design an scheme which benefited only two users through pure multiplexing gain. However, through the proposed strategy, we have designed an scheme which exploited both the global caching and multiplexing gains such that all the three users could take advantage from each transmission. 

Finally, let us discuss the achievable pair $(M,T_C)=(\frac{1}{3},1)$, where we need to adopt a different strategy. Assume we divide each of files $A$, $B$ and $C$ into three equal parts and fill the caches as follows:
\begin{eqnarray}
\nonumber Z_1&=&[A_1+B_1+C_1] \\ \nonumber
Z_2&=&[A_2+B_2+C_2] \\ \nonumber
Z_3&=&[A_3+B_3+C_3].
\end{eqnarray}
Consequently, the servers transmit the following vectors:
\begin{eqnarray}
\nonumber \mathbf{X}_1&=&\frac{\mathbf{h}_3^{\perp}}{\mathbf{h}_1.\mathbf{h}_3^{\perp}}B_1 + \frac{\mathbf{h}_2^{\perp}}{\mathbf{h}_1.\mathbf{h}_2^{\perp}}C_1 \\ \nonumber
\mathbf{X}_2&=&\frac{\mathbf{h}_3^{\perp}}{\mathbf{h}_2.\mathbf{h}_3^{\perp}}A_2 + \frac{\mathbf{h}_1^{\perp}}{\mathbf{h}_2.\mathbf{h}_1^{\perp}}C_2 \\ 
\mathbf{X}_3&=&\frac{\mathbf{h}_2^{\perp}}{\mathbf{h}_3.\mathbf{h}_2^{\perp}}A_3 + \frac{\mathbf{h}_1^{\perp}}{\mathbf{h}_3.\mathbf{h}_1^{\perp}}B_3.
\end{eqnarray}
It can be easily verified that the first user receives $A_2$, $A_3$, and $B_1+C_1$. So, with the help of its cache content, it can decode the whole file $A$. Similarly, the other users can decode their requested files. As each block $\mathbf{X}_i$ is a $2$-by-$\frac{F}{3m}$ matrix of symbols, the total delay required to fulfill the users' demands is  $T_C=\frac{F}{m}$ time slots.
\end{examp}

Example \ref{Examp_Linear_2}, also, discusses the coding delay for a linear network with four users. The details of the coding strategy of Example \ref{Examp_Linear_2}, which are provided at Appendices B and C, further clarify the basic ideas behind the proposed scheme. However, in the rest of this section, we provide the formal proof of Theorem \ref{Thm_linear}.
\\ \\
\textbf{Cache Placement Strategy:} The cache content placement phase is identical to \cite{Maddah-Ali_Fundamental_2014}: Define $t\triangleq MK/N$, and divide each file into ${K \choose t}$ non-overlapping sub-files as\footnote{It should be noted that the definition of sub-files and mini-files here differs from that of flexible networks.}:
\begin{eqnarray}
\nonumber W_n=\left(W_{n,\tau}: \tau \subset [K], |\tau|=t\right), n=1,\dots,N,
\end{eqnarray}
where each sub-file consists of $F/{K \choose t}$ bits. In the first phase, we store the sub-file $W_{n,\tau}$ in the cache of user $k$ if $k \in \tau$. Therefore, the total amount of cache each user needs for this placement is:
\begin{eqnarray}
\nonumber N\frac{F}{{K \choose t}}{K-1 \choose t-1} = MF
\end{eqnarray}
bits.

We further divide each sub-file into ${K-t-1 \choose L-1}$ non-overlapping equal-sized mini-files as follows:
\begin{eqnarray}
\nonumber W_{n,\tau}=\left(W_{n,\tau}^j: j=1,\dots,{K-t-1 \choose L-1}\right).
\end{eqnarray}
Thus, each mini-file consists of $F/\left({K \choose t} {K-t-1 \choose L-1}\right)$ bits.
\\
\\
\textbf{Content Delivery Strategy:} Consider an arbitrary $(t+L)$-subset of users denoted by $S$ (i.e. $S \subseteq [K], |S|=t+L$). For this specific subset $S$  denote all $(t+1)$-subsets of $S$ by $T_i, i=1,\dots,{t+L \choose t+1}$ (i.e. $T_i \subseteq S, |T_i|=t+1$). First, we assign a $L$-by-$1$ vector $\mathbf{u}_{S}^{T_i}$ to each $T_i$ such that
\begin{eqnarray}\label{Eq_Vector_Constraints}
\nonumber \mathbf{u}_{S}^{T_i} &\perp& \mathbf{h}_j \hspace{5mm} \mathrm{for \hspace{2mm} all} \hspace{5mm} j\in S \backslash T_i \\ 
\mathbf{u}_{S}^{T_i} &\not \perp& \mathbf{h}_j \hspace{5mm} \mathrm{for \hspace{2mm} all} \hspace{5mm} j\in T_i.
\end{eqnarray}
The following lemma specifies the required field size such that the aforementioned condition is met with high probability:
\begin{lem}
If the elements of the network transfer matrix $\mathbf{H}$ are uniformly and independently chosen from $\mathbb{F}_q$, then we can find vectors which satisfy (\ref{Eq_Vector_Constraints}) with high probability if:
\begin{equation}
q \gg (t+1) {K \choose t+L}{t+L \choose t+1}.
\end{equation}
\end{lem}
\begin{proof}
First, since the set $S \backslash T$ has $L-1$ elements, we require $\mathbf{u}_{S}^{T_i}$ to be orthogonal to $L-1$ arbitrary vectors, which is feasible in an $L$ dimensional space of any field size.

Second, the total number of non-orthogonality constraints in (\ref{Eq_Vector_Constraints}) for all possible subsets $S$ is $(t+1) {K \choose t+L}{t+L \choose t+1}$. On the other hand, it can be easily verified that the probability that two uniformly chosen random vectors in $\mathbb{F}_q$ are orthogonal is $1/q$. Thus, by using the union bound, the probability that at least one non-orthogonality constraint in (\ref{Eq_Vector_Constraints}) is violated is upper bounded by
\begin{eqnarray}
\nonumber \frac{(t+1) {K \choose t+L}{t+L \choose t+1}}{q} \ll 1,
\end{eqnarray}
which concludes the proof.
\end{proof}

For each $T_i$ define:
\begin{equation}\label{Eq_Linear_Proof_G_T}
G(T_i)=L_{r \in T_i}\left(W_{d_r,T_i \backslash \{r\}}^j\right),
\end{equation}
where $W_{d_r,T_i \backslash \{r\}}^j$ is a mini-file which is available in the cache of all users in $T_i$, except $r$, and is required by user $r$. Also $L_{r \in T_i}$ represents a random linear combination of the corresponding mini-files for all $r \in T_i$. Note that the index $j$ is chosen such that such mini-files have not been observed in the previous $(t+L)$-subsets. Thus, if we define $N(r,T \backslash \{r\})$ as the index of the next fresh mini-file required by user $r$, which is present in the cache of users $T \backslash \{r\}$, then we can rewrite:
\begin{equation}\label{Eq_Linear_General_GTi_1}
G(T_i)=L_{r \in T_i}\left(W_{d_r,T_i \backslash \{r\}}^{N(r,T_i \backslash \{r\})}\right),
\end{equation} 
Subsequently, we make the following definition for such $(t+L)$-subset $S$:
\begin{equation}\label{Eq_Linear_Proof_X_S}
\mathbf{X}(S)=\sum_{T \subseteq S, |T|=t+1}{\mathbf{u}_{S}^{T}G(T)}.
\end{equation}
We repeat the above procedure ${t+L-1 \choose t}$ times for the given $(t+L)$-subset $S$ in order to derive different independent versions of $\mathbf{X}_\omega(S), \omega=1,\dots,{t+L-1 \choose t}$. In other words, $\mathbf{X}_\omega(S)$'s only differ in the random coefficients chosen for calculating the linear combinations in (\ref{Eq_Linear_General_GTi_1}), which makes them independent linear combinations of the corresponding mini-files, with high probability. Thus, to distinguish between these different versions notationally we define:
\begin{equation}\label{Eq_Linear_General_GTi_2}
G_\omega(T_i)=L_{r \in T_i}^\omega\left(W_{d_r,T_i \backslash \{r\}}^{N(r,T_i \backslash \{r\})}\right), \mathbf{X}_\omega(S)=\sum_{T \subseteq S, |T|=t+1}{\mathbf{u}_{S}^{T}G_\omega(T)}.
\end{equation} 

Subsequently, for this $(t+L)$-subset $S$, the servers transmit the block
\begin{equation}
\left[\mathbf{X}_1(S),\dots,\mathbf{X}_{{t+L-1 \choose t}}(S)\right], 
\end{equation}
and we update $N(r, T \backslash \{r\})$ for those mini-files which have appeared in the linear combinations in (\ref{Eq_Linear_General_GTi_1}). When the above procedure for this specific subset $S$ is completed, we consider another $(t+L)$-subset of users and do the above procedure for that subset, and repeat this process until all $(t+L)$-subsets of $[K]$ have been taken into account.

Next, let us calculate the coding delay of this scheme, after which we prove the correctness of this content delivery strategy. For a fixed $(t+L)$-subset $S$ each $\mathbf{X}_\omega(S)$ is a $L$-by-$\frac{F/m}{{K \choose t} {K-t-1 \choose L-1}}$ block of symbols. Thus, the transmit block for $S$, i.e. $\left[\mathbf{X}_1(S),\dots,\mathbf{X}_{{t+L-1 \choose t}}(S)\right]$, is a $L$-by-$\frac{F/m}{{K \choose t} {K-t-1 \choose L-1}} {t+L-1 \choose t}$ block. Since this transmission should be repeated for all ${K \choose t+L}$ $(t+L)$-subsets of users, the whole transmit block size will be 
\begin{eqnarray}
\nonumber L \mathrm{-by-} \frac{ {t+L-1 \choose t}}{{K \choose t} {K-t-1 \choose L-1}} {K \choose t+L}\frac{F}{m}= L \mathrm{-by-} \frac{K(1-M/N)}{L+MK/N} \frac{F}{m},
\end{eqnarray}
which will result in the coding delay of
\begin{equation}
T_C=\frac{K(1-M/N)}{L+MK/N}\frac{F}{m}
\end{equation} 
time slots. Algorithm 2 shows the pseudo-code of the aforementioned procedure for linear networks. \\

\textbf{Correctness Proof}: Suppose the user $k$, who is interested in acquiring the file $W_{d_k}$. This file is partitioned into two parts: 1- The part already cached in this user at the first phase and constitutes of sub-files:
\begin{eqnarray}
\left(W_{d_k,\tau}: \tau \subseteq [K], |\tau|=t, k \in \tau\right).
\end{eqnarray}
2- Those parts which should be delivered to this user through the content delivery strategy, which constitutes of sub-files:
\begin{eqnarray}
\left(W_{d_k,\tau}: \tau \subseteq [K], |\tau|=t, k \not \in \tau\right).
\end{eqnarray}
Thus, since due to the following Lemma \ref{Lem_Linear_Proof_2}, the sub-files in the second category are successfully delivered to this user through the content delivery strategy, this user will decode the requested file. Moreover, since this user was arbitrarily chosen, all users will similarly decode their requested files. 

Before proving Lemma \ref{Lem_Linear_Proof_2} we need another lemma which is proved first:
\begin{lem}\label{Lem_Linear_Proof_1}
Suppose an arbitrary subset $T \subseteq [K]$ such that $|T|=t+1$, and $k \in T$. Then, through the above content placement and delivery strategy, user $k$ will be able to decode the sub-file $W_{d_k, T \backslash \{k\}}$.
\end{lem}

\begin{proof}
Consider those transmissions which are assigned to the $(t+L)$-subsets which contain $T$. There exist ${K-t-1 \choose L-1}$ of such subsets. Let us focus on one of them, namely $S$. Corresponding to $S$, the following transmit block is sent by the servers:
\begin{eqnarray}
\left[\mathbf{X}_1(S),\dots,\mathbf{X}_{{t+L-1 \choose t}}(S)\right],
\end{eqnarray}
and subsequently, user $k$ receives:
\begin{eqnarray}\label{Eq_Linear_Proof_Recieve_k_1}
\mathbf{h}_k .\left[\mathbf{X}_1(S),\dots,\mathbf{X}_{{t+L-1 \choose t}}(S)\right].
\end{eqnarray}
Let's focus on $\mathbf{h}_k .\mathbf{X}_1(S)$:
\begin{eqnarray}\label{Eq_Linear_Proof_Recieve_k_2}
\nonumber \mathbf{h}_k .\mathbf{X}_1(S)&\stackrel{(a)}=&\mathbf{h}_k . \sum_{T \subseteq S, |T|=t+1}{\mathbf{u}_{S}^{T}G_1(T)} \\ \nonumber
&\stackrel{(b)}=&\sum_{T \subseteq S, |T|=t+1, k \in T}{\left(\mathbf{h}_k .\mathbf{u}_{S}^{T}\right)G_1(T)} \\ 
&\stackrel{(c)}=&\sum_{T \subseteq S, |T|=t+1, k \in T}{\left(\mathbf{h}_k .\mathbf{u}_{S}^{T}\right)L_{r \in T}^1(W_{d_r,T \backslash \{r\}}^j)},
\end{eqnarray}
where (a) follows from (\ref{Eq_Linear_Proof_X_S}), (b) follows from the fact that 
\begin{eqnarray}
\mathbf{u}_{S}^{T} &\perp& \mathbf{h}_k \hspace{5mm} \mathrm{for \hspace{2mm} all} \hspace{5mm} k\in S \backslash T,
\end{eqnarray}
and (c) is due to (\ref{Eq_Linear_Proof_G_T}). In (\ref{Eq_Linear_Proof_Recieve_k_2}), user $k$ can extract $W_{d_k,T \backslash \{k\}}^j$ from the linear combination $L_{r \in T}^1(W_{d_r,T \backslash \{r\}}^j)$, since all the other interference terms are present at his cache. Thus, by removing interference terms, user $k$ can carve the following linear combination from (\ref{Eq_Linear_Proof_Recieve_k_2}):
\begin{eqnarray}
\nonumber L_{T \subseteq S, |T|=t+1, k \in T}^1\left(W_{d_k,T \backslash \{k\}}^j\right),
\end{eqnarray}
which is a random linear combination of ${t+L-1 \choose t}$ mini-files desired by user $k$. However, since in (\ref{Eq_Linear_Proof_Recieve_k_1}) user $k$ receives ${t+L-1 \choose t}$ independent random linear combinations of these mini-files, he can recover the whole set of mini-files:
\begin{eqnarray}
\nonumber \left(W_{d_k,T \backslash \{k\}}^j: T \subseteq S, |T|=t+1, k \in T\right).
\end{eqnarray}
Thus, for the $T$ specified in this lemma, he can recover the mini-file $W_{d_k,T \backslash \{k\}}^j$. Now, since there exist a total of ${K-t-1 \choose L-1}$ $(t+L)$-subsets containing this specific $T$, by considering the transmissions corresponding to each, this user will recover ${K-t-1 \choose L-1}$ \emph{distinct} mini-files of form $W_{d_k,T \backslash \{k\}}^j$. The distinctness is guaranteed by the appropriate updating of the index $N(\cdot,\cdot)$. These mini-files will recover the sub-file $W_{d_k,T \backslash \{k\}}$ and the proof is concluded.

\end{proof}

\begin{lem}\label{Lem_Linear_Proof_2}
Through the above content delivery strategy an arbitrary user $k$ will be able to decode all the sub-files:
\begin{eqnarray}
\left(W_{d_k,\tau}: \tau \subseteq [K], |\tau|=t, k \not \in \tau\right).
\end{eqnarray}
\end{lem}
\begin{proof}
Consider an arbitrary $\tau \subseteq [K]$ such that $|\tau|=t, k \not \in \tau$. Define $T=\tau \cup \{k\}$. Then, since to Lemma \ref{Lem_Linear_Proof_1}, user $k$ is able to decode $W_{d_k,\tau}$. Since $\tau$ was chosen arbitrarily, the proof is complete. 
\end{proof}

\begin{algorithm}\label{Alg_Main}
\caption{Multi-Server Coded Caching - Linear Networks}
\begin{algorithmic}[1]
\Procedure{PLACEMENT}{$W_1,\dots,W_N$}
\State $t \gets MK/N$
\ForAll{$n \in [N]$}
\State split $W_n$ into $(W_{n,\tau}: \tau \subset [K], |\tau|=t)$ of equal size
\ForAll{$\tau \subset [K], |\tau|=t$}
\State split $W_{n,\tau}$ into $(W_{n,\tau}^j: j=1,\dots,{K-t-1 \choose L-1})$ of equal size
\EndFor
\EndFor
\ForAll{$k \in [K]$}
\State $Z_k \gets (W_{n,\tau}^j: \tau \subset [K], |\tau|=t, k \in \tau, j=1,\dots,{K-t-1 \choose L-1}, n \in [N])$
\EndFor
\EndProcedure
\\
\Procedure{DELIVERY}{$W_1,\dots,W_N$, $d_1,\dots,d_K$}
\State $t \gets MK/N$
\ForAll{$T \subseteq [K], |T|=t+1$}
\ForAll{$r \in T$}
\State $N(r,T\backslash\{r\}) \gets 1 $
\EndFor
\EndFor
\ForAll{$S \subseteq [K], |S|=t+L$}
\ForAll{$T \subseteq S, |T|=t+1$}
\State Design $\mathbf{u}_{S}^{T}$ such that: for all $j \in S$, $\mathbf{h}_j \perp  \mathbf{u}_{S}^{T}$ if $j \not\in T$ and $\mathbf{h}_j \not \perp  \mathbf{u}_{S}^{T}$ if $j \in T$

\EndFor

\ForAll {$\omega=1,\dots,{t+L-1 \choose t}$}
\ForAll{$T \subseteq S, |T|=t+1$}
\State $G_\omega(T) \gets L_{r \in T}^\omega\left( W_{{d_r},T\backslash\{r\}}^{N(r,T\backslash\{r\})}\right)$
\EndFor
\State $\mathbf{X}_\omega(S) \gets \sum_{T \subseteq S, |T|=t+1} {\mathbf{u}_{S}^{T} G_\omega(T)}$
\EndFor

\State \textbf{transmit} $\mathbf{X}(S)=\left[\mathbf{X}_1(S),\dots,\mathbf{X}_{{t+L-1 \choose t}}(S)\right]$

\ForAll{$T \subseteq S, |T|=t+1$}
\ForAll{$r \in T$}
\State $N(r,T\backslash\{r\}) \gets N(r,T\backslash\{r\}) + 1$
\EndFor
\EndFor

\EndFor
\EndProcedure

\end{algorithmic}
\end{algorithm}

\section{Conclusions}\label{Sec_Conclusions}

In this paper, we investigated coded caching in a multi-server network where servers are connected to multiple cache-enabled clients. Based on the topology of the network, we defined three types of networks, namely, dedicated, flexible, and linear networks. In dedicated and flexible networks, we assume that the internal nodes are aware of the network topology, and accordingly route the data. In linear networks, we assume no topology knowledge at internal nodes, and thus, internal nodes perform random linear network coding. We have shown that knowledge of type of network topology plays a key role in design of proper caching mechanisms in such networks. Our results show that all network types can benefit from both caching and multiplexing gains. In fact, in dedicated and linear networks the global caching and multiplexing gains appear in additive form. However, in flexible networks they appear in multiplicative form, leading to an order-optimal solution in terms of coding delay. 

\newpage

\bibliographystyle{ieeetr}

\newpage

\section*{Appendix A: Converse Proof}\label{Appendix_Converse}

\begin{figure}
\begin{center}
\includegraphics[width=0.6\textwidth]{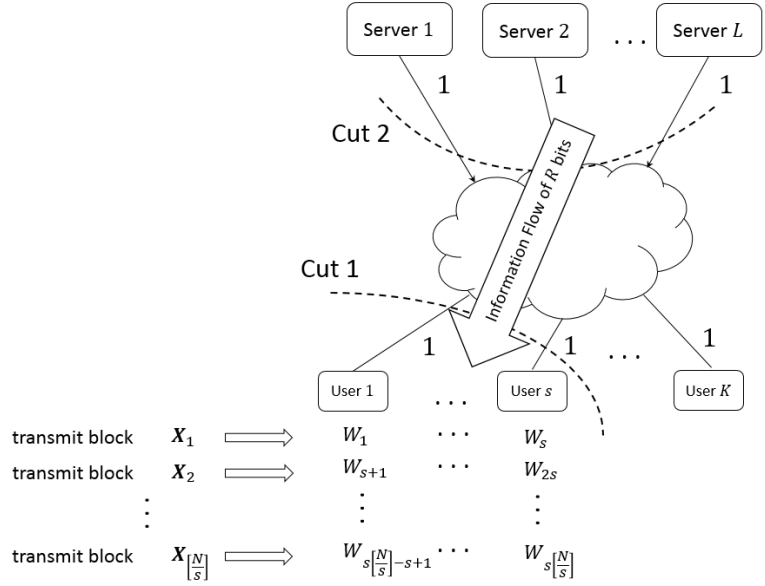}
\end{center}
\caption{Converse Proof. \label{Fig_Converse_Proof}}
\end{figure}

The proof is similar to the cut-set method presented in \cite{Maddah-Ali_Fundamental_2014}. See Fig. \ref{Fig_Converse_Proof} and let us concentrate on the first $s$ users. Define $\mathbf{X}_1$ to be the transmit block sent by the servers such that these users, with the help of their cache contents $Z_1,\dots,Z_s$, will be able to decode $W_1,\dots,W_s$. Also, define $\mathbf{X}_2$ to be the block which enables the users to decode $W_{s+1},\dots,W_{2s}$, and continue the same process such that $\mathbf{X}_{\lfloor N/s \rfloor}$ is the block which enables the users to decode $W_{s\lfloor N/s \rfloor -s+1},\dots,W_{s\lfloor N/s \rfloor}$. Also, define $R$ to be the maximum information needed to pass through the two cuts shown in the figure, by each transmit block transmission. Then we will have:
\begin{eqnarray}
\nonumber s\lfloor \frac{N}{s} \rfloor F \leq \lfloor \frac{N}{s} \rfloor R+sMF,
\end{eqnarray}
which will result in 
\begin{eqnarray}
\nonumber R \geq \left(s-\frac{s}{ \lfloor \frac{N}{s} \rfloor}M\right)F.
\end{eqnarray}
However, we have:
\begin{eqnarray}
\nonumber	D^*(M) &\geq& \frac{R}{\mathrm{min-cut}} \\ \nonumber
				&\geq&  \frac{R}{\min(s,L)m} \\ 
				&\geq& \frac{1}{\min(s,L)} \left(s-\frac{s}{ \lfloor \frac{N}{s} \rfloor}M\right)\frac{F}{m}.
\end{eqnarray}
Now we can maximize on the free parameter $s$ to arrive at the tightest bound, which concludes the proof.

\section*{Appendix B: Details of Example \ref{Examp_Linear_2} ($L=2,N=4,K=4$)}\label{Appendix_Two_Servers}

In this appendix, we consider the scenario in Example \ref{Examp_Linear_2} for the case of two servers. For each memory size $M=0,\dots,4$, we present the scheme which achieves the coding delay as stated in Example \ref{Examp_Linear_2}.
\begin{itemize}
\item $M=0$

In this case, we do not have any cache space available at the users. Suppose we divide each file into three equal-sized non-overlapping parts:
\begin{eqnarray}
\nonumber A&=&[A^1,A^2,A^3] \\ \nonumber
B&=&[B^1,B^2,B^3] \\ \nonumber
C&=&[C^1,C^2,C^3] \\ \nonumber
D&=&[D^1,D^2,D^3].
\end{eqnarray}
Then, the servers transmit the following blocks, in sequence:
\begin{eqnarray}\label{Eq_L2_K4_N4_M0_transmit}
\nonumber \mathbf{X}(\{1,2\})&=& \mathbf{h}_1^{\perp}B^1+\mathbf{h}_2^{\perp}A^1 \\ \nonumber
\mathbf{X}(\{1,3\})&=& \mathbf{h}_1^{\perp}C^1+\mathbf{h}_3^{\perp}A^2 \\ \nonumber
\mathbf{X}(\{1,4\})&=& \mathbf{h}_1^{\perp}D^1+\mathbf{h}_4^{\perp}A^3 \\ \nonumber
\mathbf{X}(\{2,3\})&=& \mathbf{h}_2^{\perp}C^2+\mathbf{h}_3^{\perp}B^2 \\ \nonumber
\mathbf{X}(\{2,4\})&=& \mathbf{h}_2^{\perp}D^2+\mathbf{h}_4^{\perp}B^3 \\ 
\mathbf{X}(\{3,4\})&=& \mathbf{h}_3^{\perp}D^3+\mathbf{h}_4^{\perp}C^3.
\end{eqnarray}
Let's focus on the first user which receives:
\begin{eqnarray}
\nonumber \mathbf{h}_1.\mathbf{X}(\{1,2\})&=& (\mathbf{h}_1.\mathbf{h}_2^{\perp})A^1 \\ \nonumber
\mathbf{h}_1.\mathbf{X}(\{1,3\})&=& (\mathbf{h}_1.\mathbf{h}_3^{\perp})A^2 \\ \nonumber
\mathbf{h}_1.\mathbf{X}(\{1,4\})&=& (\mathbf{h}_1.\mathbf{h}_4^{\perp})A^3. 
\end{eqnarray}
From the above data, this user can recover the whole file $A$. Similarly, other users can decode their requested files.

The transmission stated in (\ref{Eq_L2_K4_N4_M0_transmit}) consists of six blocks of size $2$-by-$\frac{F}{3m}$, resulting in a coding delay of $T_C=6 \frac{F}{3m}=2\frac{F}{m}$.

\item $M=1$

Consider the cache content placement used in \cite{Maddah-Ali_Fundamental_2014}: First divide each file into $4$ equal-sized non-overlapping sub-files:
\begin{eqnarray}
\nonumber A&=&[A_1,A_2,A_3,A_4] \\ \nonumber
B&=&[B_1,B_2,B_3,B_4] \\ \nonumber
C&=&[C_1,C_2,C_3,C_4] \\ \nonumber
D&=&[D_1,D_2,D_3,D_4],
\end{eqnarray}
and then, fill the caches as follows:
\begin{eqnarray}
\nonumber Z_1&=&[A_1,B_1,C_1,D_1] \\ \nonumber
Z_2&=&[A_2,B_2,C_2,D_2] \\ \nonumber
Z_3&=&[A_3,B_3,C_3,D_3] \\ \nonumber
Z_4&=&[A_4,B_4,C_4,D_4].
\end{eqnarray}
Such placement respects the memory constraint of $M=1$. Also, divide each sub-file into two equal parts of size $\frac{1}{2}\frac{F}{4}=\frac{F}{8}$ bits:
\begin{eqnarray}
\nonumber A_i&=&[A_i^1,A_i^2],  \\ \nonumber
B_i&=&[B_i^1,B_i^2],  \\ \nonumber
C_i&=&[C_i^1,C_i^2],  \\ \nonumber
D_i&=&[D_i^1,D_i^2], 
\end{eqnarray} 
where $i=1,2,3,4$. In the second phase, we send the following blocks of size $2$-by-$\frac{F}{4}$ bits:
\begin{eqnarray}\label{Eq_L2_K4_N4_M1_transmit}
\nonumber \mathbf{X}(\{1,2,3\})&=&[\mathbf{h}_1^{\perp} L_{\{2,3\}}^1(B_3^1,C_2^1)+\mathbf{h}_2^{\perp} L_{\{1,3\}}^1(A_3^1,C_1^1)+\mathbf{h}_3^{\perp} L_{\{1,2\}}^1(A_2^1,B_1^1), \\ \nonumber
& &\mathbf{h}_1^{\perp} L_{\{2,3\}}^2(B_3^1,C_2^1)+\mathbf{h}_2^{\perp} L_{\{1,3\}}^2(A_3^1,C_1^1)+\mathbf{h}_3^{\perp} L_{\{1,2\}}^2(A_2^1,B_1^1)] \\ \nonumber
\mathbf{X}(\{1,2,4\})&=&[\mathbf{h}_1^{\perp} L_{\{2,4\}}^1(B_4^1,D_2^1)+\mathbf{h}_2^{\perp} L_{\{1.4\}}^1(A_4^1,D_1^1)+\mathbf{h}_4^{\perp} L_{\{1,2\}}^1(A_2^2,B_1^2), \\ \nonumber
& &\mathbf{h}_1^{\perp} L_{\{2,4\}}^2(B_4^1,D_2^1)+\mathbf{h}_2^{\perp} L_{\{1,4\}}^2(A_4^1,D_1^1)+\mathbf{h}_4^{\perp} L_{\{1,2\}}^2(A_2^2,B_1^2)] \\ \nonumber
\mathbf{X}(\{1,3,4\})&=&[\mathbf{h}_1^{\perp}L_{\{3,4\}}^1(C_4^1,D_3^1)+\mathbf{h}_3^{\perp}L_{\{1,4\}}^1(A_4^2,D_1^2)+\mathbf{h}_4^{\perp}L_{\{1,3\}}^1(A_3^2,C_1^2), \\ \nonumber
& &\mathbf{h}_1^{\perp}L_{\{3,4\}}^2(C_4^1,D_3^1)+\mathbf{h}_3^{\perp}L_{\{1,4\}}^2(A_4^2,D_1^2)+\mathbf{h}_4^{\perp}L_{\{1,3\}}^2(A_3^2,C_1^2)] \\ \nonumber
\mathbf{X}(\{2,3,4\})&=&[\mathbf{h}_2^{\perp}L_{\{3,4\}}^1(C_4^2,D_3^2)+\mathbf{h}_3^{\perp}L_{\{2,4\}}^1(B_4^2,D_2^2)+\mathbf{h}_4^{\perp}L_{\{2,3\}}^1(B_3^2,C_2^2), \\ \nonumber
& &\mathbf{h}_2^{\perp}L_{\{3,4\}}^2(C_4^2,D_3^2)+\mathbf{h}_3^{\perp}L_{\{2,4\}}^2(B_4^2,D_2^2)+\mathbf{h}_4^{\perp}L_{\{2,3\}}^2(B_3^2,C_2^2)]. \\
\end{eqnarray}
Let's focus on the first user. From the above transmissions he recovers:
\begin{eqnarray}
\nonumber \mathbf{h}_1 . \mathbf{X}(\{1,2,3\})&=&[(\mathbf{h}_1.\mathbf{h}_2^{\perp}) L_{\{1,3\}}^1(A_3^1,C_1^1)+(\mathbf{h}_1.\mathbf{h}_3^{\perp}) L_{\{1,2\}}^1(A_2^1,B_1^1), \\ \nonumber
& &(\mathbf{h}_1.\mathbf{h}_2^{\perp}) L_{\{1,3\}}^2(A_3^1,C_1^1)+(\mathbf{h}_1.\mathbf{h}_3^{\perp}) L_{\{1,2\}}^2(A_2^1,B_1^1)] \\ \nonumber
&=& [L^1(A_3^1,C_1^1,A_2^1,B_1^1),L^2(A_3^1,C_1^1,A_2^1,B_1^1)] \\ \nonumber
\mathbf{h}_1 . \mathbf{X}(\{1,2,4\})&=&[(\mathbf{h}_1.\mathbf{h}_2^{\perp}) L_{\{1,4\}}^1(A_4^1,D_1^1)+(\mathbf{h}_1.\mathbf{h}_4^{\perp}) L_{\{1,2\}}^1(A_2^2,B_1^2), \\ \nonumber
& &(\mathbf{h}_1.\mathbf{h}_2^{\perp}) L_{\{1,4\}}^2(A_4^1,D_1^1)+(\mathbf{h}_1.\mathbf{h}_4^{\perp}) L_{\{1,2\}}^2(A_2^2,B_1^2)] \\ \nonumber
&=&[L^1(A_4^1,D_1^1,A_2^2,B_1^2),L^2(A_4^1,D_1^1,A_2^2,B_1^2)] \\ \nonumber
\mathbf{h}_1 .\mathbf{X}(\{1,3,4\})&=&[(\mathbf{h}_1 .\mathbf{h}_3^{\perp})L_{\{1,4\}}^1(A_4^2,D_1^2)+(\mathbf{h}_1 .\mathbf{h}_4^{\perp})L_{\{1,3\}}^1(A_3^2,C_1^2), \\ \nonumber
& &(\mathbf{h}_1 .\mathbf{h}_3^{\perp})L_{\{1,4\}}^2(A_4^2,D_1^2)+(\mathbf{h}_1 .\mathbf{h}_4^{\perp})L_{\{1,3\}}^2(A_3^2,C_1^2)] \\ \nonumber
&=& [L^1(A_4^2,D_1^2,A_3^2,C_1^2),L^2(A_4^2,D_1^2,A_3^2,C_1^2)]. \\
\end{eqnarray}
(Although user 1 also receives $\mathbf{h}_1 .\mathbf{X}(\{2,3,4\})$, such information is of no value to him.)
With the help of its cache contents the first user can eliminate the undesired terms and obtain:
\begin{eqnarray}
\nonumber & &[L(A_3^1,A_2^1),L'(A_3^1,A_2^1)] \rightarrow A_3^1,A_2^1 \\ \nonumber
& &[L(A_4^1,A_2^2),L'(A_4^1,A_2^2)] \rightarrow A_4^1,A_2^2 \\ \nonumber
& & [L(A_4^2,A_3^2),L'(A_4^2,A_3^2)] \rightarrow A_4^2,A_3^2.
\end{eqnarray}
Since $A_1^1$ and $A_1^2$ is already available in first user's cache location, he can subsequently recover the whole block $A$. Similarly, all other users can recover their requested files.

The transmission scheme adopted in (\ref{Eq_L2_K4_N4_M1_transmit}) consists of four $2$-by-$\frac{F}{4m}$ blocks which will result in the coding delay $T_C=4 \frac{F}{4m}=\frac{F}{m}$ time slots.

\item $M=2$

Consider the cache content placement used in \cite{Maddah-Ali_Fundamental_2014}: First divide each file into $6$ equal-sized non-overlapping sub-files:
\begin{eqnarray}
\nonumber A&=&[A_1,A_2,A_3,A_4,A_5,A_6] \\ \nonumber
B&=&[B_1,B_2,B_3,B_4,B_5,B_6] \\ \nonumber
C&=&[C_1,C_2,C_3,C_4,C_5,C_6] \\ \nonumber
D&=&[D_1,D_2,D_3,D_4,D_5,D_6],
\end{eqnarray}
and then, fill the caches as follows:
\begin{eqnarray}
\nonumber Z_1&=&[A_1,A_2,A_3,B_1,B_2,B_3,C_1,C_2,C_3,D_1,D_2,D_3] \\ \nonumber
Z_2&=&[A_1,A_4,A_5,B_1,B_4,B_5,C_1,C_4,C_5,D_1,D_4,D_5] \\ \nonumber
Z_3&=&[A_2,A_4,A_6,B_2,B_4,B_6,C_2,C_4,C_6,D_2,D_4,D_6] \\ \nonumber
Z_4&=&[A_3,A_5,A_6,B_3,B_5,B_6,C_3,C_5,C_6,D_3,D_5,D_6].
\end{eqnarray}
In the second phase, we send the following block of symbols of size $2$-by-$\frac{F}{2m}$:
\begin{eqnarray}\label{Eq_L2_K4_N4_M2_transmit}
\nonumber \mathbf{X}=[& &\mathbf{h}_1^{\perp} L_{\{2,3,4\}}^1(B_6,C_5,D_4)+\mathbf{h}_2^{\perp}L_{\{1,3,4\}}^1(A_6,C_3,D_2)+\mathbf{h}_3^{\perp}L_{\{1,2,4\}}^1(A_5,B_3,D_1)+\mathbf{h}_4^{\perp}L_{\{1,2,3\}}^1(A_4,B_2,C_1), \\ \nonumber
& &\mathbf{h}_1^{\perp} L_{\{2,3,4\}}^2(B_6,C_5,D_4)+\mathbf{h}_2^{\perp}L_{\{1,3,4\}}^2(A_6,C_3,D_2)+\mathbf{h}_3^{\perp}L_{\{1,2,4\}}^2(A_5,B_3,D_1)+\mathbf{h}_4^{\perp}L_{\{1,2,3\}}^2(A_4,B_2,C_1),\\ \nonumber
& &\mathbf{h}_1^{\perp} L_{\{2,3,4\}}^3(B_6,C_5,D_4)+\mathbf{h}_2^{\perp}L_{\{1,3,4\}}^3(A_6,C_3,D_2)+\mathbf{h}_3^{\perp}L_{\{1,2,4\}}^3(A_5,B_3,D_1)+\mathbf{h}_4^{\perp}L_{\{1,2,3\}}^3(A_4,B_2,C_1) \hspace{6mm}]. \\
\end{eqnarray}
Let's focus on the first user who receives:
\begin{eqnarray}
\nonumber [& &L^1(A_4,A_5,A_6,B_2,B_3,C_1,C_3,D_1,D_2), \\ \nonumber
& &L^2(A_4,A_5,A_6,B_2,B_3,C_1,C_3,D_1,D_2),\\ \nonumber
& &L^3(A_4,A_5,A_6,B_2,B_3,C_1,C_3,D_1,D_2) \hspace{6mm}].
\end{eqnarray}
This user also has the unwanted terms $B_2,B_3,C_1,C_3,D_1,D_2$ in his cache, and after removing them from above linear combinations he has three different linear combinations of its required terms $A_4$, $A_5$, and $A_6$. After solving these equations, and with the help of $A_1$, $A_2$, and $A_3$ stored in his cache,  he can recover the whole file $A$. Similarly the other users are able to decode their required files.

The transmit block stated in (\ref{Eq_L2_K4_N4_M2_transmit}) is of size $2$-by-$\frac{F}{2m}$ vector, resulting in $T_C=\frac{1}{2}\frac{F}{m}$ time slots.

\item $M=3$

In this case, by the scheme proposed in \cite{Maddah-Ali_Fundamental_2014}, all four users can get useful information through a single transmission from a single server. Thus, we cannot further reduce the delay by activating the other server. Thus, by activating just one server and based on \cite{Maddah-Ali_Fundamental_2014} a coding delay of $T_C=\frac{1}{4}\frac{F}{m}$ time slots is obtained.

\item $M=4$

In the case of $M=4$, all four files can be stored in the cache of each user, and the required delivery delay in the second phase is zero $T_C=0$.
\end{itemize}

\section*{Appendix C: Details of Example \ref{Examp_Linear_2} ($L=3,N=4,K=4$)}\label{Appendix_Three_Servers}

In this example, we consider the three server case in Example \ref{Examp_Linear_2}, and for all values of $M=0,\dots,4$ present the schemes that lead to achievable rates.

\begin{itemize}
\item $M=0$

In this case, we do not have any cache space available at the user locations. Suppose we divide each file into three equal-sized non-overlapping parts:
\begin{eqnarray}
\nonumber A&=&[A^1,A^2,A^3] \\ \nonumber
B&=&[B^1,B^2,B^3] \\ \nonumber
C&=&[C^1,C^2,C^3] \\ \nonumber
D&=&[D^1,D^2,D^3].
\end{eqnarray}
The three servers can then send the following $3$-by-$1$ vectors:
\begin{eqnarray}\label{Eq_L3_K4_N4_M0_transmit}
\nonumber \mathbf{X}(\{1,2,3\})&=& \mathbf{u}_{\{1,2,3\}}^{\{1\}}A^1+\mathbf{u}_{\{1,2,3\}}^{\{2\}}B^1+\mathbf{u}_{\{1,2,3\}}^{\{3\}}C^1\\ \nonumber
\mathbf{X}(\{1,2,4\})&=&  \mathbf{u}_{\{1,2,4\}}^{\{1\}}A^2+\mathbf{u}_{\{1,2,4\}}^{\{2\}}B^2+\mathbf{u}_{\{1,2,4\}}^{\{4\}}D^1\\ \nonumber
\mathbf{X}(\{1,3,4\})&=& \mathbf{u}_{\{1,3,4\}}^{\{1\}}A^3+\mathbf{u}_{\{1,3,4\}}^{\{3\}}C^2+\mathbf{u}_{\{1,3,4\}}^{\{4\}}D^2\\ 
\mathbf{X}(\{2,3,4\})&=& \mathbf{u}_{\{2,3,4\}}^{\{2\}}B^3+\mathbf{u}_{\{2,3,4\}}^{\{3\}}C^3+\mathbf{u}_{\{2,3,4\}}^{\{4\}}D^3,
\end{eqnarray}
where we require
\begin{eqnarray}
\nonumber &&\mathbf{u}_S^T \perp \mathbf{h}_j, \forall \mathbf{h}_j \in S \backslash T \\ 
&&\mathbf{u}_S^T \not \perp \mathbf{h}_j, \forall \mathbf{h}_j \in T.
\end{eqnarray}
In this example, since we have three dimensional transmit vectors (three servers) and $|S \backslash T|=2$, such vectors can be found.

Let's focus on the first user who receives:
\begin{eqnarray}
\nonumber \mathbf{h}_1.\mathbf{X}(\{1,2,3\})&=& \left(\mathbf{h}_1.\mathbf{u}_{\{1,2,3\}}^{\{1\}}\right)A^1+\left(\mathbf{h}_1.\mathbf{u}_{\{1,2,3\}}^{\{2\}}\right)B^1+\left(\mathbf{h}_1.\mathbf{u}_{\{1,2,3\}}^{\{3\}}\right)C^1=\left(\mathbf{h}_1.\mathbf{u}_{\{1,2,3\}}^{\{1\}}\right)A^1\\ \nonumber
\mathbf{h}_1.\mathbf{X}(\{1,2,4\})&=&  \left(\mathbf{h}_1.\mathbf{u}_{\{1,2,4\}}^{\{1\}}\right)A^2+\left(\mathbf{h}_1.\mathbf{u}_{\{1,2,4\}}^{\{2\}}\right)B^2+\left(\mathbf{h}_1.\mathbf{u}_{\{1,2,4\}}^{\{4\}}\right)D^1=\left(\mathbf{h}_1.\mathbf{u}_{\{1,2,4\}}^{\{1\}}\right)A^2\\ 
\mathbf{h}_1.\mathbf{X}(\{1,3,4\})&=& \left(\mathbf{h}_1.\mathbf{u}_{\{1,3,4\}}^{\{1\}}\right)A^3+\left(\mathbf{h}_1.\mathbf{u}_{\{1,3,4\}}^{\{3\}}\right)C^2+\left(\mathbf{h}_1.\mathbf{u}_{\{1,3,4\}}^{\{4\}}\right)D^2=\left(\mathbf{h}_1.\mathbf{u}_{\{1,3,4\}}^{\{1\}}\right)A^3.
\end{eqnarray}
The first user can then successfully decode its requested file. Similarly, the other users will also be able to decode their requested files. 

The transmission stated in (\ref{Eq_L3_K4_N4_M0_transmit}) consists of four $3$-by-$\frac{F}{3m}$ blocks, resulting in  $T_C=\frac{4F}{3m}$ time slots.

\item $M=1$

The cache content placement is the same as \cite{Maddah-Ali_Fundamental_2014}. Then, the transmit block by the three servers is:
\begin{equation}\label{Eq_L3_K4_N4_M1_transmit}
\mathbf{X}=[\mathbf{X}_1,\mathbf{X}_2, \mathbf{X}_3],
\end{equation}
where (for $\omega=1,2,3$)
\begin{eqnarray}
\nonumber \mathbf{X}_\omega&=&\mathbf{u}_{\{1,2,3,4\}}^{\{1,2\}}L_{\{1,2\}}^\omega(A_2,B_1)+\mathbf{u}_{\{1,2,3,4\}}^{\{1,3\}}L_{\{1,3\}}^\omega(A_3,C_1)+\mathbf{u}_{\{1,2,3,4\}}^{\{1,4\}}L_{\{1,4\}}^\omega(A_4,D_1) \\ 
&+&\mathbf{u}_{\{1,2,3,4\}}^{\{2,3\}}L_{\{2,3\}}^\omega(B_3,C_2)+\mathbf{u}_{\{1,2,3,4\}}^{\{2,4\}}L_{\{2,4\}}^\omega(B_4,D_2)+\mathbf{u}_{\{1,2,3,4\}}^{\{3,4\}}L_{\{3,4\}}^\omega(C_4,D_3).
\end{eqnarray}

Now let's focus on the first user who receives:
\begin{equation}
y_1=\mathbf{h}_1 . \mathbf{X}=[\mathbf{h}_1 .\mathbf{X}_1,\mathbf{h}_1 .\mathbf{X}_2, \mathbf{h}_1 .\mathbf{X}_3].
\end{equation}
Let's consider first the term:
\begin{eqnarray}
\nonumber \mathbf{h}_1 .\mathbf{X}_1&=&\left(\mathbf{h}_1 .\mathbf{u}_{\{1,2,3,4\}}^{\{1,2\}}\right)L_{\{1,2\}}^1(A_2,B_1)+\left(\mathbf{h}_1 .\mathbf{u}_{\{1,2,3,4\}}^{\{1,3\}}\right)L_{\{1,3\}}^1(A_3,C_1)+\left(\mathbf{h}_1 .\mathbf{u}_{\{1,2,3,4\}}^{\{1,4\}}\right)L_{\{1,4\}}^1(A_4,D_1) \\ \nonumber
&+&\left(\mathbf{h}_1 .\mathbf{u}_{\{1,2,3,4\}}^{\{2,3\}}\right)L_{\{2,3\}}^1(B_3,C_2)+\left(\mathbf{h}_1 .\mathbf{u}_{\{1,2,3,4\}}^{\{2,4\}}\right)L_{\{2,4\}}^1(B_4,D_2)+\left(\mathbf{h}_1 .\mathbf{u}_{\{1,2,3,4\}}^{\{3,4\}}\right)L_{\{3,4\}}^1(C_4,D_3)\\ \nonumber
&=&\left(\mathbf{h}_1 .\mathbf{u}_{\{1,2,3,4\}}^{\{1,2\}}\right)L_{\{1,2\}}^1(A_2,B_1)+\left(\mathbf{h}_1 .\mathbf{u}_{\{1,2,3,4\}}^{\{1,3\}}\right)L_{\{1,3\}}^1(A_3,C_1)+\left(\mathbf{h}_1 .\mathbf{u}_{\{1,2,3,4\}}^{\{1,4\}}\right)L_{\{1,4\}}^1(A_4,D_1) \\ 
&=& L^1(A_2,A_3,A_4,C_1,B_1,D_1).
\end{eqnarray}
As this user has cached $B_1,C_1,D_1$ in the first phase, it can remove these terms from this linear combination to obtain
\begin{eqnarray}
\nonumber L(A_2,A_3,A_4).
\end{eqnarray}
Thus, user $1$ can recover a linear combination of its requested sub-files from $\mathbf{h}_1 .\mathbf{X}_1$. From, $\mathbf{h}_1 .\mathbf{X}_2$ and $\mathbf{h}_1 .\mathbf{X}_3$ he can obtain two other independent linear combinations from which he can recover all three subfiles $A_2,A_3,A_4$. Since he already has $A_1$ in his cache, he can decode the whole $A$ file. Similarly, all the other users can also decode their requested files.

The transmit block stated in (\ref{Eq_L3_K4_N4_M1_transmit}) consists of one $3$-by-$\frac{3F}{4m}$ vectors, resulting in $T_C=\frac{3}{4}\frac{F}{m}$ time slots.

\item $M=2$
In this case, we only activate two of the servers and thus the problem reduces to the case with $L=2,N=4,K=4$ for which we achieved $T_C=\frac{1}{2}\frac{F}{m}$.

\item $M=3$
In this case, we only activate one server and thus the problem reduces to \cite{Maddah-Ali_Fundamental_2014} with $T_C=\frac{1}{4}\frac{F}{m}$.

\item $M=4$
In this case we have $T_C=0$.
\end{itemize}

\end{document}